\documentclass[aps,pra,groupedaddress,showpacs,twocolumn,superscriptaddress]{revtex4}

\usepackage{amsmath}
\usepackage{amssymb}
\usepackage{amsthm}
\usepackage{graphics}
\usepackage{graphicx}

\newtheorem{theorem}{Theorem}
\newtheorem{lemma}[theorem]{Lemma}
\newtheorem{corollary}[theorem]{Corollary}

\newtheorem{proposition}[theorem]{Proposition}

\newcommand{\ket}[1]{|#1\rangle}               %ket
\newcommand{\bra}[1]{\langle #1|}              %bra
\newcommand{\dyad}[2]{\ket{#1}\bra{#2}}        %dyad
\newcommand{\Tr}{{\rm Tr}}                     %trace

\newcommand{\ii}{\mathrm{i}}

\newcommand{\Eq}[1]{Eq.~\eqref{#1}}

\newcommand{\AC}{\mathcal{A}}

\newcommand{\CC}{\mathcal{C}}

\newcommand{\EC}{\mathcal{E}}
\newcommand{\FC}{\mathcal{F}}
\newcommand{\GC}{\mathcal{G}}
\newcommand{\HC}{\mathcal{H}}

\newcommand{\LC}{\mathcal{L}}

\newcommand{\PC}{\mathcal{P}}

\newcommand{\RC}{\mathcal{R}}
\newcommand{\SC}{\mathcal{S}}
\newcommand{\TC}{\mathcal{T}}
\newcommand{\UC}{\mathcal{U}}
\newcommand{\VC}{\mathcal{V}}

\newcommand{\CP}{\mathrm{CP}}
\newcommand{\ZZ}{\mathbb{Z}}

\def\outl#1{\par{\medskip\noindent\hspace*{.5cm}\bf  \mathversion{bold}#1\mathversion{normal}\smallskip} }  \def\xa{} \def\xb{}  

\def\outl#1{}  \def\xa{} \def\xb{}  
\long\def\ca#1\cb{} %Use for commenting out: \ca...\cb

\begin{document}

\title{Tripartite Entanglement in Qudit Stabilizer States and Application in Quantum Error Correction}

\author{Shiang Yong Looi}
\email{slooi@andrew.cmu.edu}
\affiliation{Department of Physics, Carnegie Mellon University, Pittsburgh,
Pennsylvania 15213, USA}

\author{Robert B. Griffiths}
\email{rgrif@andrew.cmu.edu}
\affiliation{Department of Physics, Carnegie Mellon University, Pittsburgh,
Pennsylvania 15213, USA}

\date{Version of 8th of July 2011}

\begin{abstract}

Consider a stabilizer state on $n$ qudits, each of dimension $D$ with $D$ being a prime or a squarefree integer, divided into three mutually disjoint sets or parts.  Generalizing a result of Bravyi et al. [J. Math. Phys. \textbf{47}, 062106 (2006)] for qubits ($D=2$), we show that up to local unitaries on the three parts the state can be written as a tensor product of unentangled single-qudit states, maximally entangled EPR pairs, and tripartite GHZ states.   We employ this result to obtain a complete characterization of the properties of a class of channels associated with stabilizer error-correcting codes, along with their complementary channels.

\end{abstract}

\pacs{03.67.Mn, 03.67.Hk}
\maketitle
\tableofcontents

\section{Introduction}
\label{section1}

\xa
\outl{Tripartite entanglement is not solved.}
\xb

The study of entangled quantum states of systems consisting of two or more
parts is a central problem in quantum information theory.  The Schmidt
decomposition provides a fairly complete characterization of the pure states
of a bipartite system. However, mixed states on bipartite systems and pure
states on systems of three or more parts present a much more difficult
problem---see \cite{quant-ph/0504163} for a comprehensive review---and a
relatively complete understanding of the situation exists only for some very
special cases.

The present paper considers the special case of (pure) stabilizer states on
$n$ qudits, each of dimension $D$, and addresses the problem of characterizing
the corresponding tripartite state when the $n$ qudits are partitioned into
three disjoint sets $A$, $B$, and $C$, and arbitrary unitary transformations
are allowed on each of the three parts. The case of qubits, $D=2$, was studied
by Bravyi et al.\ \cite{quantph.0504208}, who showed that such a stabilizer
state is equivalent, up to local unitaries on the three parts, to a tensor
product of pure unentangled single qubit states; maximally entangled two-qubit
states or EPR pairs, with one qubit in one part and the other qubit in a
different part; and GHZ states on three qubits, one lying in each part.  In
this paper we generalize these results to the case $D > 2$, where $D$ is
either a prime or a squarefree integer (i.e., not divisible by the square of
any integer greater than 1).

\xa
\outl{Introduction to stabilizer formalism and graph states}
\xb

The stabilizer formalism \cite{PhysRevA.54.1862, quantph.9705052,  QuantumComputation} was first introduced to simplify the construction and analysis of quantum error correction codes. Soon thereafter it was generalized from qubits to higher dimensional qudits \cite{quantph.9608048,   quantph.0005008}. Most of the codes known when the formalism was introduced, and the majority of those discovered since, are stabilizer codes. The formalism has also been used for measurement-based quantum computation \cite{PhysRevA.68.022312} and fault-tolerant topological quantum computation \cite{quant-ph/9707021}. There has been a lot of research on single-qubit local unitary (LU) and single-qubit local Clifford (LC) equivalence of qubit stabilizer states/graph states \cite{PhysRevA.72.014307, PhysRevA.71.062323, PhysRevA.75.032325} but here we consider partitionings where each part can have several qudits and arbitrary gates acting on qudits belonging to the same part are permitted.

\xa
\outl{Previous work on stabilizer code channels and subset information groups}
\xb

In \cite{PhysRevA.81.032326} we studied a class of channels obtained from qudit stabilizer (equivalently, additive graph) codes where a subset of the carrier qudits is lost.  We fully characterized their information carrying capacities in terms of \emph{subset information groups}, a concept related to the notion of correctable algebras introduced in \cite{PhysRevA.76.042303}. We also provided an efficient algorithm to find the subset information group. In this paper we adopt the name \emph{stabilizer code channels} for such channels.% and apply the tripartite result with the help of Choi-Jamio\l kowski isomorphism or map-state duality to obtain a duality relation between the subset information group of a stabilizer code channel and that of its complementary channel (i.e., the channel going to the ``environment''). 

\xa
\outl{Overview of paper.}
\xb

The paper is organized as follows: Section~\ref{section2} introduces various
concepts that will be used later: Pauli and Clifford operators, one- and
two-qudit gates, stabilizer and graph states.  It also contains some
mathematical results, one of which, Corollary \ref{corollary5}, is of some
interest by itself: it allows the decomposition of stabilizer states into a
tensor product of such states when $D=d_1d_2 \cdots$ is a product of mutually coprime factors.  In the following Section \ref{section3} we prove that any bipartite stabilizer state in the case of squarefree $D$ is equivalent, up to unitaries on the two parts, to a collection of unentangled single-qubit states and maximally entangled EPR pairs. This could have been studied using the Schmidt decomposition, but the techniques used here are also needed in the following section.

The central result of this paper is the tripartition Theorem~\ref{theorem7} stated and proved in Section \ref{section4}. It shows that when $D$ is squarefree a stabilizer state on three parts can be decomposed into a tensor product of single-qudit states, two-qudit EPR pairs and three-qudit GHZ states. With the help of Choi-Jamio\l kowski isomorphism or map-state duality, this result is applied in Section \ref{section5} to the stabilizer code channels where we show they can always be decomposed into a product of a perfect quantum channel, a perfectly decohering channel, and a depolarizing channel (not all of which need be present). We also prove that the subset information groups corresponding to a stabilizer code channel and its complementary channel obey a duality relation, in that one completely specifies the other. While the results are specific to stabilizer code channels, we show that they can also be used to provide bounds on channel capacities for some other cases.

Section \ref{conclusion} summarizes our findings and suggests some directions
for future research. 

% ==============================================================================

\section{Preliminary Concepts and Definitions}
\label{section2}

\subsection{Qudit Pauli Operators}
\label{subsection2a}

\xa
\outl{Qudit Pauli operators and its properties.}
\xb

Most of the following preliminary concepts have been introduced in \cite{PhysRevA.81.032326}, \cite{PhysRevA.78.042303} and we present them here again for completeness. We generalize the notion of Pauli operators to higher dimensional Hilbert spaces where $D \geqslant 2$. The $X$ and $Z$ Pauli operators are defined in the computational basis as
\begin{align}
\label{eqn1}
Z=\sum_{j=0}^{D-1}\omega^j\dyad{j}{j},\quad X=\sum_{j=0}^{D-1}\dyad{j}{j+1},
\end{align}
and they satisfy
\begin{align}
\label{eqn2}
X^D=Z^D=I,\quad XZ=\omega ZX,\quad \omega = \mathrm{e}^{2 \pi \ii /D},
\end{align}
where the addition of integers in Eq.~\eqref{eqn1} is modulo $D$. For a collection of $n$ qudits we use subscripts to identify the corresponding Pauli operators unless otherwise stated: thus $Z_i$ and $X_i$ operate on the space of qudit $i$. The Hilbert spaces of individual qudits are denoted by $\HC_i$, and that of $n$ qudits by $\HC^{\otimes n} := \HC_1 \otimes \HC_2 \otimes \cdots \otimes \HC_n$. Operators of the form
\begin{equation}
\label{eqn3}
\lambda^{\gamma}X_1^{x_1}Z_1^{z_1}\otimes X_2^{x_2}Z_2^{z_2}\otimes\cdots
\otimes X_n^{x_n}Z_n^{z_n}
\end{equation} 
will be referred to as \emph{Pauli products}, where $\lambda := \mathrm{e}^{2 \pi \mathrm{i}/(2D)}$ (so $\lambda^2 = \omega$) and $\gamma$ is an integer in $\ZZ_{2D}$, the ring of integers modulo $2D$. For a fixed $n$, the collection of all possible Pauli products in Eq.~\eqref{eqn3} forms a group under operator multiplication, the \emph{Pauli group} $\PC_n$. 

For every $p \in \PC_n$, $p^D$ is either $I$ or $-I$. The \emph{order} of a Pauli product $p \in \PC_n$ is defined as the smallest integer $1 \leqslant \alpha \leqslant D$ such that $p^\alpha \propto I$. Our definition of order is nonstandard in that we only require the power of the Pauli products to be \emph{proportional} to the identity. Note that the order of any Pauli product must divide $D$.

While $\PC_n$ is not abelian, it has the property that any two elements \emph{commute up to a phase}: $p_1p_2 = \omega^{\alpha_{12}} p_2p_1$, with $\alpha_{12}$ an integer in $\ZZ_D$ that depends on $p_1$ and $p_2$. One can find subgroups of $\PC_n$ that are abelian, for example the set of Pauli products with only powers of $Z$ on every qudit.

\begin{proposition}
\label{proposition1}
Let $\AC $ be set of mutually commuting Pauli products in $\PC_n$ (for example, abelian subgroups of $\PC_n$). Then $\AC$ can have at most $D^n$ linearly independent elements. 
\end{proposition}
\begin{proof}
The elements of $\AC$ can be viewed as $D^n$-by-$D^n$ matrices. Then it is impossible to simultaneously diagonalize $D^n + 1$ or more mutually commuting \emph{and} linearly independent $D^n$-by-$D^n$ matrices.
\end{proof}

\xa
\outl{Pauli operators as an orthonormal basis.}
\xb

The collection of $D^{2n}$ Pauli products in Eq.~\eqref{eqn3} with $\gamma=0$, i.e. a pre-factor of $1$, forms an orthonormal basis of $\LC(\HC^{\otimes n})$, the space of linear operators on $\HC^{\otimes n}$, with respect to the Hilbert-Schmidt inner product
\begin{equation}
\label{eqn4}
\frac{1}{D^n}\Tr\{q_1^\dagger q_2^{}\} = \delta_{q_1,q_2}, \; \forall q_1,q_2 \in \PC_n \text{ with pre-factor of 1.}
\end{equation}

%There is a bijective map between $\QC_n$ and the quotient group $\PC_n /\{\omega^{\alpha}{I}\}$ for $\alpha\in\ZZ_D$ where $\{\omega^{\alpha}{I}\}$, the center of $\PC_n$, consists of phases multiplying the identity operator on $n$ qudits.

% Why bijection and not isomorphism? \QC_n can be viewed as an abelian group that is isomorhic to the quotient group provided we redefine group multiplication. If we stick to matrix multiplication, then multiplication of elements in \QC_n can result in a element outside of \QC_n

\subsection{Single-Qudit and Two-Qudit Clifford Operators}
\label{subsection2b}

\xa
\outl{Three single-qudit Clifford operators.}
\xb

Having defined Pauli operators, we now generalize other single-qubit and two-qubit operators to $D \geqslant 2$. The qudit generalization of the Hadamard gate is the
\emph{Fourier gate}
\begin{equation}
\label{eqn5}
F := \frac{1}{\sqrt{D}}\sum_{j=0}^{D-1}\omega^{jk}\dyad{j}{k}.
\end{equation}

For an invertible integer $\alpha \in \ZZ_D$ (i.e. integer for which there exists $\bar \alpha \in \ZZ_D$ such that $\alpha \bar \alpha \equiv 1 \bmod D$), we define a \emph{multiplicative gate}
\begin{equation}
\label{eqn6}
S^{(\alpha)}:=\sum_{j=0}^{D-1}\dyad{j}{\alpha j}.
\end{equation}
The requirement that $\alpha$ be invertible ensures that $S^{(\alpha)}$ is unitary. (For $D=2$, the only invertible integer is $\alpha=1$, hence $S^{(\alpha)}$ is just the identity.)

Next we define the \emph{phase gate} as
\begin{align}
\label{eqn7}
W := \left\{
\begin{array}{ll}
\sum_{j=0}^{D-1} \lambda^{-j(j+2)} \ket{j} \bra{j} & \quad \text{if $D$ is even} \\
\sum_{j=0}^{D-1} \lambda^{-j(j+1)} \ket{j} \bra{j} & \quad \text{if $D$ is odd.}
\end{array}\right.
\end{align}
where $\lambda = \mathrm{e}^{2 \pi \ii / (2D)}$. The phase gate was first studied by Nielsen et al.\  in \cite{PhysRevA.66.022317} for the general $D$ case.

\xa
\outl{Action of these single-qudit operators on Paulis.}
\xb

The three single-qudit operators defined above as well as the Pauli operators defined in Eq.~\eqref{eqn1} are examples of Clifford unitaries, by which we mean unitaries that map Pauli products to Pauli products under conjugation. For instance, $FZF^\dag  = X$ and $FXF^\dag  = Z^{-1}$. The results of conjugating the Pauli operators by $F$, $S^{(\alpha)}$ and $W$ are summarized in Table \ref{table1}.

\begin{table}
\begin{tabular}{|c|c|c|c|}
\hline
Pauli operator    & $\;F\;$ & $\;S^{(\alpha)}\;$ & $\;W\;$ \\
\hline
\hline
$Z$ & $X$ & $Z^\alpha$ & $Z$\\
\hline
$X$ & $\;Z^{-1}\;$ & $X^{\bar \alpha}$ & $\;\lambda XZ\;$ (even $D$)\;\\
 & & & $\;XZ\;$ (odd $D$) \\
\hline
\end{tabular}
\caption{The result of conjugation of Pauli operators by one-qudit gates $F, S^{(\alpha)}$ and $W$. ($\bar \alpha$ is the multiplicative inverse of $\alpha$ mod $D$ and $\lambda = \mathrm{e}^{2 \pi \ii / (2D)}$.)}
\label{table1}
\end{table}
%=\mathrm{e}^{2 \pi \ii / (2D)

\xa
\outl{Two-qudit Clifford operators.}
\xb

%%RBG alterations below

The generalizations to $D\geqslant2$  of CP and CNOT gates are the Clifford unitaries
\begin{equation}
\label{eqn8}
\CP_{12} = \sum_{j=0}^{D-1}\dyad{j}{j}_1 \otimes Z^j_2 = \sum_{j,k=0}^{D-1} \omega^{jk}\dyad{j}{j}_1 \otimes \dyad{k}{k}_2
\end{equation}
and
\begin{equation}
\label{eqn9}
\mathrm{CNOT}_{12} := \sum_{j=0}^{D-1}\dyad{j}{j}_1 \otimes X_2^j = \sum_{j,k=0}^{D-1}\dyad{j}{j}_1 \otimes \dyad{k}{k+j}_2,
\end{equation}
where qudit 1 is the control while qudit 2 is the target. The CP and CNOT gates are related by a local Fourier gate defined in Eq.~\eqref{eqn5}, similar to the $D=2$ case,
\begin{equation}
\label{eqn10}
\mathrm{CNOT}_{12}=(I_1 \otimes F_2) \CP_{12} (I_1 \otimes F_2)^\dag.
\end{equation}

\begin{proposition}
\label{proposition2}
For $D$ prime, let $p \in \PC_n$ be a Pauli product on $n$ qudits [\Eq{eqn3}] and assume that $p$ is not the identity on qudit 1, i.e. $x_1 \neq 0$ or $z_1 \neq 0$ or both. Then there exists a Clifford unitary $U$ such that $UpU^\dag \propto X_1 I_2 \cdots I_n$. Further if $p^D = I$, then it is possible to have $UpU^\dag = X_1 I_2 \cdots I_n$.
\end{proposition}
\begin{proof}
If $x_1 = 0$, then conjugate $p$ by the Fourier gate, $F$ so that $x_1 \neq 0$. Then transform $X_1^{x_1} Z_1^{z_1}$ to $X_1^{x_1}$ by conjugating it with the $W$ gate a sufficient number of times. Next use the $S^{(\alpha)}$ gate to produce $X_1$. See Table \ref{table1} for the result of these conjugations. (Note that we relied on the fact that $\ZZ_D$ is a field when $D$ is prime in the last two operations. The more general result for arbitrary $D$ is studied in \cite{PhysRevA.66.022317}.)
If $p$ is now the identity on all the other qudits $i=2,3,\ldots, n$ we are
done. Otherwise, for each non-identity qudit, set the Pauli operator to $X$
employing the procedure above. If at this point $p$ is not identity on qudit
$2$, i.e. $p = X_1 X_2 \cdots$, then $X_2$ can be changed to $I_2$ by
performing $\mathrm{CNOT}_{12}$. This is repeated where needed so that the
Pauli product is the identity on every qudit except qudit 1, which proves
$UpU^\dag$ is proportional to $X_1$. If now $p^D = I$, any remaining phase is
necessarily some power of $\omega$, and can be removed by conjugation with
powers of $Z_1$.
\end{proof}

\subsection{Stabilizer Codes and States. Partitions}
\label{subsection2c}

\xa
\outl{Definition of stabilizer codes and states}
\xb

Let $\SC \subset \PC_n$ be an abelian subgroup consisting of linearly
independent Pauli products. Then $|\SC|$ must divide $D^n$ and $s^D = I$ for
all $ s \in \SC$. (For prime $D$, every element except the identity is
necessarily of order $D$ so $|\SC|$ is always a power of $D$.)
Given $\SC$, define the set of states, $\CC := \{ \ket{\psi} \in \HC^{\otimes
  n} : s\ket{\psi} = \ket{\psi}, \; \forall \; s \in \SC\}$. It is easy to
check that $\CC$ forms a linear space, which we call the
\emph{stabilizer code}, with $\SC$ its \emph{stabilizer group} %
\footnote{Note that all elements of $\SC$ leave \emph{each} element of the
  subspace $\CC$ unchanged.  The larger subgroup that maps $\CC$ into itself
  without the requirement that each $\ket{\psi}$ in $\CC$ be mapped to itself
  could also be called its ``stabilizer,'' but we are not using ``stabilizer''
  in this second sense.}. %

In \cite{PhysRevA.78.042303} it was shown that $\CC$ and $\SC$ are dual in
the sense that one completely specifies the other and they satisfy the
relation $|\SC| \times \dim(\CC) = D^n$; $\dim(\CC)$ is the dimension of
$\CC$. In quantum error correction literature, if $\dim(\CC) = D^k$ for some
integer $0 \leqslant k \leqslant n$, then it is customary to write $\CC = [[n,
k]]_D$ because one can think of encoding $k$ qudits in the $D^k$-dimensional
subspace contained in the space of $n$ carrier qudits.
Let $U$ be a Clifford unitary and $\SC$ a stabilizer group with $\CC$ being
its corresponding stabilizer code. Then $\SC' := U \SC U^\dag = \{ UsU^\dag :
s \in \SC\}$ is also a stabilizer group stabilizing the code $\CC' = \{ U
\ket{\psi} :  \ket{\psi} \in \CC\}$. For a detailed review on Clifford
unitaries and stabilizer states for arbitrary $D$ see \cite{PhysRevA.71.042315}.

If $|\SC| = D^n$, then $\SC$ stabilizes a unique state and we call it the
\emph{stabilizer state}, denoted by $\ket{\SC}$. The projector onto the state
can be written as a sum of elements in $\SC$, as shown in \cite{quantph.0406168}
\begin{align}
\label{eqn11}
\ket{\SC} \bra{\SC} = \frac{1}{D^n} \sum_{s \in \SC} s.
\end{align}

\xa
\outl{Two examples of stabilizer states.}
\xb

Two simple examples of stabilizer states are the EPR pair and the GHZ state,
expressed below for any $D \geqslant 2$ with their respective stabilizer
groups,
\begin{align}
\label{eqn12}
\ket{\text{EPR}}_{12} &= \frac{1}{\sqrt{D}} \sum_{i=0}^{D-1} \ket{i}_{1} \ket{i}_{2} \nonumber,\\
\SC &= \langle X_1 X_2, Z_1 Z_2^{-1} \rangle
\end{align}
and
\begin{align}
\label{eqn13}
\ket{\text{GHZ}}_{123} &= \frac{1}{\sqrt{D}} \sum_{i=0}^{D-1} \ket{i}_{1} \ket{i}_{2} \ket{i}_{3},\nonumber \\
\SC &= \langle X_1 X_2 X_3, Z_1 Z_2^{-1} , Z_1 Z_3^{-1} \rangle
\end{align}
where the angular brackets denote the group generated by products of the elements in the list.

\xa
\outl{Generators of stabilizer group.}
\xb

When $D$ is prime, $\SC$ can always be generated by $n$ suitably chosen group
elements, $\SC = \langle s_1, s_2, \ldots, s_n \rangle$ such that the order of
each $s_i$ is $D$. For non-prime $D$, one might need more than $n$ generators is some cases. We call these group elements \emph{stabilizer generators}
or generators. Note that the set of generators is not unique -- there are many distinct choices of generators that generate the same group, for example $\SC = \langle s_1 s_2, s_2, \ldots, s_n \rangle$.

\begin{proposition}
\label{proposition3}
Let $\SC$ be a stabilizer group with $D^n$ elements where $D$ is prime. Let $\TC = \langle t_1, t_2, \ldots, t_m \rangle$ be a subgroup of $\SC$ with $D^m$ elements where $1 \leqslant m < n$. Then there exists a set of $n-m$ elements, $\{t_{m+1}, \ldots, t_n\} \subset \SC$ such that $\SC = \langle t_1, t_2, \ldots, t_m,t_{m+1}, \ldots, t_n  \rangle$. 
\end{proposition}
\begin{proof}
First pick an element of $\SC$ not in $\TC$ and call it $t_{m+1}$. Since $D$ is prime, the order of $t_{m+1}$ must be $D$. Then the set $\{ t'\: t_{m+1}^{\alpha} | \; t' \in \langle t_1, t_2, \ldots, t_m \rangle, \alpha \in \ZZ_D \},  \equiv \langle t_1, t_2, \ldots, t_m, t_{m+1} \rangle $ is a subgroup of $\SC$ with $D^{m+1}$ elements. Repeat this incremental addition of generators until the set of generators generates $\SC$.
\end{proof}

\xa
\outl{Bipartitions of stabilizer state.}
\xb

A stabilizer state $\ket{\SC} \in \HC^{\otimes n}$ naturally ``lives'' in a tensor product space of $n$ qudits but one can imagine a coarser-grained partitioning where the $n$ qudits are divided into two parts, labeled $A$ and $B$, which we will call a bipartition. One can regard any state on the total Hilbert space as an entangled state on $\HC_A \otimes \HC_B$ and we will also refer to such a state as a bipartition. Tripartitions and generalizations to higher number of partitions can be analogously defined. (Obviously partitions can be defined on any multi-partite state, not just stabilizer states.)

\xa
\outl{Reduced density operators of stabilizer state.}
\xb

A useful expression for reduced density operators of multipartite
stabilizer states is the following. For a bipartite stabilizer state
$\ket{\SC} \in \HC_A \otimes \HC_B$ let
\begin{align}
\label{eqn14}
\SC_A := \left\{ s \in \SC : \Tr_B\{s\} \neq 0 \right\}
\end{align}
be the elements of $\SC$ equal to the identity on $\HC_B$.  They form a
subgroup of $\SC$ (e.g., \cite{PhysRevA.81.032326}), and in light of 
\Eq{eqn11},
\begin{align}
\label{eqn15}
\rho_A = \Tr_B\{\ket{\SC}\bra{\SC}\} = \frac{1}{D^{n_A}} \sum_{s \in \SC_A} s.
\end{align}
If we square both sides we see that $\rho_A^2 = (|\SC_A|/D^{n_A}) \rho_A$,
which means that the reduced density operator of a stabilizer state has
identical positive eigenvalues, so it is proportional to a
projector. Additionally it satisfies
\begin{align}
\label{eqn16}
\text{rank} (\rho_A) = \frac{D^{n_A}}{|\SC_A|}.
\end{align}
Therefore, $\rho_A$ is proportional to the identity if and only if the subgroup $\SC_A$ has only the identity element.

\xa
\outl{Tensor product of stabilizer state.}
\xb

Finally let $\ket{\SC}$ and $\ket{\TC}$ be stabilizer states on distinct sets
of $n$ and $m$ qudits with stabilizer group $\SC = \langle s_1, \ldots, s_n
\rangle$ and $\TC = \langle t_1, \ldots, t_m \rangle$. Then clearly the tensor
product of these states $\ket{\VC} = \ket{\SC} \otimes \ket{\TC}$ is also a
stabilizer state with the stabilizer group
\begin{align}
\label{eqn17}
\VC &= \langle s_1, \ldots, s_n \rangle \otimes \langle t_1,
\ldots, t_m \rangle \nonumber \\ &= \langle s_1 \otimes I , \ldots, s_n
\otimes I, I \otimes t_1, \ldots, I \otimes t_m \rangle.
\end{align}
Conversely, given a stabilizer state $\ket{\VC} \in \HC^{\otimes
n} \otimes \HC^{\otimes m}$, if the stabilizer group can be written as a
tensor product of two stabilizer groups, $\VC = \langle s_1, \ldots, s_n
\rangle \otimes \langle t_1, \ldots, t_m \rangle$, then $\ket{\VC} = \ket{\SC}
\otimes \ket{\TC}$, since $\ket{\SC}$ and $\ket{\TC}$ are uniquely determined
by their respective stabilizer groups.

\subsection{Decomposition of Stabilizer States of Composite Dimensions}
\label{subsection2d}

Let the integer $D$ have the prime decomposition
\begin{align}
\label{eqn18}
D = p_1^{\epsilon_1} p_2^{\epsilon_2} \cdots p_m^{\epsilon_m},
\end{align}
where the $p_i$ are distinct primes and the $\epsilon_i$ positive integers. The
following theorem is useful when the qudit dimension $D$ is composite.

\begin{theorem}[Chinese Remainder Decomposition of Stabilizer State]
\label{theorem4}

Let $\AC$ be an abelian group of linearly independent Pauli products on
$n$ qudits, each of dimension $D$, and let \Eq{eqn18} be the prime
decomposition of $D$. 

 Then $\AC$ is unitarily equivalent to a tensor
product of $m$ abelian groups in the sense that
\begin{align}
\label{eqn19}
(\UC \otimes \cdots \otimes \UC) \AC (\UC \otimes \cdots \otimes \UC)^\dag =
\bigotimes_{i=1}^m \AC_{i},
\end{align} 
where $\UC$ is a unitary acting on the $D$-dimensional space of a single
qudit,
each $\AC_{i}$  is an abelian group of linearly
independent Pauli products on $n$ qudits of dimension $p_i^{\epsilon_i}$,
and
\begin{align}
\label{eqn20}
| \AC |  = | \AC_{i} |\cdot| \AC_{2} | \cdots | \AC_{m} | .
\end{align}
\end{theorem}

The proof is in Appendix A. As stabilizer groups are examples of such abelian
groups, one has:

\begin{corollary}
\label{corollary5}
Let $\ket{\SC}$ be a stabilizer state on $n$ qudits of dimension $D$, prime
decomposition given by \Eq{eqn18}.

Then there exists a single-qudit unitary $\UC$ such that
\begin{align}
\label{eqn21}
\UC \otimes \cdots \otimes \UC \ket{\SC} = \bigotimes_{i=1}^M \ket{\SC_{i}}
\end{align}
where each $\ket{\SC_{i}}$ is a stabilizer state on $n$ qudits of dimension $p_i^{\epsilon_i}$.
\end{corollary}

\begin{proof}
Let $\SC$ denote the stabilizer group of $\ket{\SC}$ which has $D^n$ linearly independent Pauli products. Then by the theorem above $\SC$, an abelian group of linearly independent Pauli products, is equivalent up to local unitaries to a tensor product of $m$ stabilizer groups of dimensions $p_1^{\epsilon_1}, p_2^{\epsilon_2}, \ldots, p_m^{\epsilon_m}$, each stabilizing its own stabilizer state $\ket{\SC_{i}}$.
\end{proof}

When applied to an arbitrary stabilizer state $\ket{\SC_6}$ on qudits of
$D=6$, this corollary states that it is equivalent up to local single-qudit
unitaries to a tensor product of two stabilizer states $\ket{\SC_2} \otimes
\ket{\SC_3}$, one on $n$ qubits and the other on $n$ qutrits. Essentially each
$D=6$ qudit has an internal tensor product structure that can be decomposed to
a qubit and a qutrit. Therefore in studies of entanglement of stabilizer
states of $D=6$, it is sufficient to just consider qubit and
qutrit stabilizer states.

In this paper, the corollary above is used to extend various results on
stabilizer states that hold for prime D to the case where $D$ is
\emph{squarefree}, meaning that $\epsilon_i=1$ in \eqref{eqn18} for every
$i$. 

\subsection{Graph States}
\label{subsection2e}

\xa
\outl{Definition of qudit graph state.}
\xb

Let $\Gamma_{ij} = \Gamma_{ji}$ be the adjacency matrix of an undirected graph
$G$ on $n$ vertices with no loops ($\Gamma_{ii}=0)$). Each $\Gamma_{jj}$, 
the \emph{weight} of the edge connecting vertices $i$ and $j$, can take any
value in  $\ZZ_D$, with (as usual) $\Gamma_{ij}=0$ in the absence of an edge.

The graph state $\ket{G}$ is a state on $n$ qudits of dimension $D$ defined as
\begin{align}
\label{eqn51}
\ket{G} := \left( \prod_{i=1}^{n-1} \prod_{j=i+1}^{n} \CP_{ij}^{\Gamma_{ij}} \right)  \ket{+}_1 \otimes \cdots \otimes \ket{+}_n 
\end{align}
where
\begin{align}
\label{eqn52}
\ket{+} := \frac{1}{\sqrt{D}} \sum_{j=0}^{D-1} \ket{j}
\end{align}
and the two-qudit gate $\CP_{ij}$ is defined in Eq.~\eqref{eqn8}. Note that
the $\CP$ gates all commute with each other so there is no need to specify the
order in which they act on the kets. For non-prime $D$, there are alternative
ways to define graph states; see, e.g., \cite{quantph.0703112}.

\xa
\outl{Stabilizer property of graph state. Generators of stabilizer group.}
\xb

All graph states are stabilizer states but the converse is not true; however it was shown in \cite{quantph.0111080} and \cite{quantph.0703112} for prime $D$ (and therefore also for squarefree $D$ by Corollary \ref{corollary5}) that all stabilizer states are equivalent up to local Clifford unitaries to graph states. It is often more convenient to work with the stabilizer group (denoted by $\SC_G$) rather than the ket itself. For a given graph state $\ket{G}$ with adjacency matrix $\Gamma$, there is a canonical set of $n$ stabilizer generators, $\{g_i\}$ given by
\begin{align}
\label{eqn53}
g_i := X_i \left( \prod_{j=1}^{n} Z_j^{-\Gamma_{ij}} \right) \qquad \text{for }i=1,2,\ldots, n
\end{align}
which of course satisfies $g_i \ket{G} = \ket{G}$ for all $g_i$, so we have $\SC_G =\langle g_1, g_2, \ldots, g_n \rangle$. These operators are called correlation operators in \cite{quantph.0602096}.

%------------------------------------------------------------------------------

\section{Bipartition of qudit stabilizer states}
\label{section3}

\xa
\outl{State known results on bipartitions of graph states.}
\xb

Entanglement across bipartitions of stabilizer states has been studied in
Section 3 of \cite{PhysRevA.69.062311} and \cite{quantph.0406168}. Here we
shall extend their result to all squarefree $D\geqslant 2$ with the theorem
below. The entanglement of a bipartite state can always be studied in terms of its Schmidt decomposition but we present an alternative approach here because it is helpful in explaining the techniques that will be used in the Tripartition Theorem in Section \ref{section4}.

\xa
\outl{Introduce concept of unentangled subsystems.}
\xb

Before stating the Bipartition Theorem let us study some simple stabilizer states to understand how \emph{unentangled subsystems} in each part can obscure the actual amount of entanglement present. We shall consider the two stabilizer states on three qudits below and ask how much entanglement is present across the $A$-$B$ bipartition:
\begin{align}
\label{eqn100}
\SC^{(1)} &= \langle Z_{A_1}^{} Z_{B_1}^{-1}, X_{A_1} X_{B_1}, X_{A_2}\rangle \nonumber \\
\SC^{(2)} &= \langle Z_{A_1}^{} Z_{B_1}^{-1}, X_{A_1}^{} Z_{A_2}^{-1} X_{B_1}^{}, Z_{A_1}^{-1} X_{A_2}^{}\rangle
\end{align}
where the subscripts $A_1, A_2$ denote qudits in part $A$ and analogously $B_1$ in part $B$.

First observe that $\SC^{(1)}$ can be factorized as $\langle Z_{A_1}^{} Z_{B_1}^{-1}, X_{A_1} X_{B_1}\rangle \otimes \langle X_{A_2}\rangle$ (see discussion at the end of Subsection \ref{subsection2c}) and the unentangled qudit $A_2$ is irrelevant as far as entanglement between $A$ and $B$ is concerned. From here it is straightforward to see that $\langle Z_{A_1} Z_{B_1}, X_{A_1} X_{B_1}\rangle$ stabilizes the EPR pair described in Eq.~\eqref{eqn12}.

In the second case, it is harder to tell how entangled the state is by just looking at the stabilizer group $\SC^{(2)}$, even though it differs only by a local unitary on part $A$ from the previous state, $\ket{\SC^{(2)}} = \CP_{A_1 A_2} \ket{\SC^{(1)}}$. This tells us there must be some hidden unentangled subsystem in part $A$ that upon removal will result in a simpler two-qudit stabilizer state, just like the first example.

\xa
\outl{Link concept of unentangled subsystems back to reduced density operators of stabilizer state. Lead into introduction of theorem.}
\xb

A systematic way to ``detect'' the presence of unentangled subsystems in stabilizer states is by inspecting the reduced density operator on each part or equivalently the subgroups $\SC_A, \SC_B$ (see Subsection \ref{subsection2c}). The Bipartition Theorem below is essentially just a formal statement that once all the unentangled subsystems are removed, all that remains is a collection of EPR pairs.

\begin{theorem}[Bipartition of stabilizer state]
\label{theorem6}
For squarefree $D$, let $\ket{\SC}$ be a stabilizer state on $n\geqslant 2$ qudits. For any bipartition of $\ket{\SC} \in \HC_A \otimes \HC_B$, there exists Clifford unitaries $U_A, U_B$ on each part such that $U_A U_B \ket{\SC}$ is a collection of maximally entangled EPR pairs and unentangled single-qudit states, i.e.
\begin{align}
\label{eqn101}
U_A U_B \ket{\SC}_{AB} = \ket{\mathrm{EPR}}_{AB}^{\otimes m_{AB}} \otimes \ket{+}_A^{\otimes m_A} \otimes \ket{+}_B^{\otimes m_B}.
\end{align}
Note that $m_A$, $m_B$ or $m_{AB}$ can be zero.
\end{theorem}

\begin{proof}
By invoking Corollary \ref{corollary5}, we can decompose $\ket{\SC}$ into
several stabilizer states where each of them is on $n$ qudits of prime
dimension. Therefore it is sufficient to prove the theorem only for prime
$D$.

If the subgroups $\SC_A$ and $\SC_B$ both contain only the identity element,
then by Eq.~\eqref{eqn15} both $\rho_A$ and $\rho_B$ are proportional to the
identity. This is equivalent to $\SC$ not containing any element that is
non-trivial only in one part, such as $X_{A_1} I_B$. This also means
$\ket{\SC}$ is maximally entangled and therefore is equivalent to a collection
of EPR pairs. 

Otherwise assume that $\SC_A$ has at least one element, $s \in \PC_n$ not
equal to the identiy, a Pauli product which acts non-trivially on at least one
qudit in part $A$. Without loss of generality we can assume that qudit is
$A_1$. Then by Proposition \ref{proposition2} we know there exists a Clifford
operation, $U_A$ such that $U^{}_A s U^{\dag}_A = X_{A_1} I_{A_2} \cdots
I_{A_{n_A}}$.

Next consider the new stabilizer group for $U_A \ket{\SC}$ and choose
$X_{A_1}$ as one of the generators so that $U^{}_A \SC U_A^\dag = \langle
X_{A_1}, s_2, \ldots, s_n \rangle$, which is always possible as shown in
Proposition \ref{proposition3}. Since $s_2$ must commute with $X_{A_1}$, there
cannot be any $Z_{A_1}$ operator in it and hence must be of the form $s_2 =
X_{A_1}^{\alpha} \otimes p_{A \backslash A_1}$, where $p_{A \backslash A_1}$
is some Pauli product on qudits $A_2, \ldots, A_{n_A}$. If $\alpha=0$, then do
nothing. Otherwise replace $s_2$ with $s'_2 := X_{A_1}^{-\alpha} s_2 = I_{A_1}
\otimes p_{A \backslash A_1}$, so that the new generator is the identity on
qudit $A_1$. This replacement does not change the group being
generated.

Repeat this procedure for all the other generators $s_3, \ldots, s_n$. In doing so we now have a new set of generators such that there is only one generator that is non-trivial on qudit $A_1$ while all other generators have identity on $A_1$. The end result is a stabilizer group that can be written as a tensor product, $U_A \SC U_A^\dag = \langle X_{A_1} \rangle \otimes \langle s'_2, \ldots, s'_n \rangle$.

Following the discussion at the end of Subsection \ref{subsection2c}, we can write $U_A \ket{\SC} = \ket{+}_{A_1} \otimes \ket{\SC'}$ where $\ket{+}$ is defined in Eq.~\eqref{eqn52} and $\ket{\SC'}$ is stabilized by $\langle s_2', \ldots, s_n'\rangle$. In other words, we have extracted an unentangled subsystem from part $A$ and are left with a stabilizer state with $n-1$ qudits. Repeat this process on both parts until both $\SC_A$ and $\SC_B$ contain only the identity element. This concludes the proof.
\end{proof}

\xa
\outl{Allude to reuse of extraction result in tripartition theorem.}
\xb

In fact this extraction of unentangled subsystems works for stabilizer states with \emph{any} number of parts since we can always view the part of interest as $A$ and all the other parts as $B$.

%===========================================================================

\section{Tripartition of qudit stabilizer states}
\label{section4}

\xa
\outl{Mention known results on tripartition of qubit graph states.}
\xb

The problem of tripartition of qubit ($D=2$) stabilizer states has been studied by Bravyi et al in \cite{quantph.0504208}. They proved that such states are always equivalent up to unitaries on each part to a collection of GHZ states, maximally entangled EPR pairs and unentangled single-qubit states; see Figure \ref{tripart_graph} for a simple illustration. In the same paper, they also provided partial solutions to the general problem with more than three parties. Here we extend their tripartition result to  squarefree $D \geqslant 2$ using a method mentioned but not used in their paper.

\begin{figure}%[ht]
\begin{center}
\includegraphics{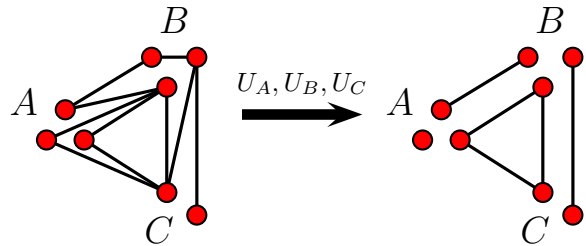}
\caption{Example of equivalence of a tripartite graph state to a collection of single-qudit states, EPR pairs and GHZ state.}
\label{tripart_graph}
\end{center}
\end{figure}

\begin{theorem}[Tripartition of stabilizer state]
\label{theorem7}
For squarefree $D$, let $\ket{\SC}$ be a stabilizer state with $n \geq 3$ qudits. For any tripartition of $\ket{\SC} \in \HC_A \otimes \HC_B \otimes \HC_C$, there exists Clifford unitaries $U_A, U_B, U_C$ on each part such that $U_A U_B U_C \ket{\SC}$ is a collection of GHZ states, maximally entangled EPR pairs and unentangled single-qudit states, i.e.
\begin{align}
\label{eqn200}
U_A U_B U_C \ket{\SC} &= \ket{\mathrm{GHZ}}^{\otimes m_{ABC}} \otimes \ket{\mathrm{EPR}}^{\otimes m_{AB}} \nonumber \\
& \quad \otimes \ket{\mathrm{EPR}}^{\otimes m_{BC}}\otimes \ket{\mathrm{EPR}}^{\otimes m_{AC}} \nonumber \\
& \quad \otimes \ket{+} ^{\otimes m_A} \otimes \ket{+} ^{\otimes m_B}\otimes \ket{+} ^{\otimes m_C}
\end{align}
where $m_A, m_B, m_C, m_{AB}, m_{BC}, m_{AC}, m_{ABC}$ are non-negative integers that can be zero. 
\end{theorem}
\begin{proof}

  It is sufficient to prove the theorem only for prime $D$ because Corollary
  \ref{corollary5} extends the proof to squarefree $D$. The proof can be
  divided into three major steps. In Step 1 we simply repeat the unentangled
  subsystem extraction procedure explained in the previous section; the
  details will not be repeated here.  As a consequence the three reduced
  density operators $\rho_A, \rho_B$ and $\rho_C$ are all proportional to the
  identity.
  In Step 2, we extract EPR pairs using local Clifford unitaries acting on two
  parts at a time.  Thus for parts $A$ and $B$, we try to find $U_A, U_B$ such
  that $U_A U_B I_C \ket{\SC} = \ket{\mathrm{EPR}}_{AB} \otimes
  \ket{\SC'}$. The identical procedure can be used to extract EPR pairs from
  $A$ and $C$, and from $B$ and $C$, so we only need to discuss how it works
  for $A$ and $B$.  Finally in Step 3 we prove that the state remaining after
  all these extractions must be a collection of GHZ states.

We begin Step 2 with the observation that the reduced density operator
$\rho_{AB}$ for the combined parts $A$ and $B$ is given by
\begin{align}
\label{eqn201}
\rho_{AB} &= \Tr_C \{\ket{\SC} \bra{\SC}\} = D^{-(n_A+n_B)} \sum_{i=1}^{|\SC_{AB}|} s^{(i)}_{AB} \nonumber \\
&= D^{-(n_A+n_B)} \sum_{i=1}^{|\SC_{AB}|} s^{(i)}_{A} \otimes s^{(i)}_{B},
\end{align}
where $\SC_{AB} := \{s \in \SC : \Tr_C\{s\} \neq 0\}$ and $s^{(i)}_{A},
s^{(i)}_{B}$ are Pauli products on $A$, $B$ respectively. There is some phase
ambiguity in the final expression because a phase on $s^{(i)}_{A}$
can be moved to $s^{(i)}_{B}$ or vice versa, but the following proof does not
depend on how the phase is assigned.

STEP 2A - If the collection $\{s^{(i)}_{A}\}$ defined in Eq.~\eqref{eqn201}
contains two elements $s^{(j)}_{A}, s^{(k)}_{A}$ that do not commute, then at
least one EPR pair can be extracted from parts $A$ and $B$ as we now show. We
can always assume those two elements satisfy the commutation relation
\begin{align}
\label{eqn202}
s^{(k)}_{A} s^{(j)}_{A} = \omega \: s^{(j)}_{A} s^{(k)}_{A},
\end{align}
because if the phase picked up is instead some higher power of $\omega$, we
can replace $s^{(k)}_A$ with an appropriate power of $s^{(k)}_A$, since $D$ is
prime. 
Now transform $s^{(j)}_{AB}$ so that it is the identity on every qudit
\emph{except} $A_1$ and $B_1$, i.e.
\begin{align}
\label{eqn203}
s^{(j)}_{AB} = Z_{A_1}^{} \otimes Z_{B_1}^{-1},
\end{align}
by applying Proposition \ref{proposition2} twice, first on part $A$ and then
on part $B$, using Clifford unitaries local to each part. The commutation
relation between $s^{(j)}_{A}, s^{(k)}_{A}$ in Eq.~\eqref{eqn202} is left
unchanged because it is invariant under conjugations by unitaries.

STEP 2B - The fact that $[s^{(j)}_{AB}, s^{(k)}_{AB}]=0$ together with
Eqs.~\eqref{eqn202} and \eqref{eqn203} imply that
\begin{align}
\label{eqn204}
s^{(k)}_{AB} &= X_{A_1}^{} Z_{A_1}^{\alpha} p^{}_{A \backslash A_1} \otimes X_{B_1}^{} Z_{B_1}^{\beta} q^{}_{B \backslash B_1}
\end{align}
where $p,q$ are some Pauli products on the remaining qudits on part $A,B$
respectively and $\alpha, \beta \in \ZZ_D$. To see this, first observe that
Eqs.~\eqref{eqn202} and \eqref{eqn203} tell us the exponent on $X_{A_1}$ is
necessarily 1 while imposing no conditions on the exponent of $Z_{A_1}$ nor
the Pauli operators on the other qudits in part $A$. Next, $s^{(j)}_{AB}$
commuting with $s^{(k)}_{AB}$ fixes the exponent of $X_{B_1}$ to be 1,
resulting in Eq.~\eqref{eqn204}.

In \Eq{eqn204} the $Z$ operators on qudits $A_1$ and $B_1$ can be removed by
conjugating $s^{(k)}_{AB}$ with the phase gate $W$ applied to these qudits
(see Table \ref{table1}). Lastly, the Pauli products $p$ and $q$ can be set to
$I$ using methods outlined in the proof of Proposition \ref{proposition2}. The
conjugations by $\mathrm{CNOT}$ gates described there with $A_1,
B_1$ being the control qudits do not modify $s^{(j)}_{AB}$ at all. At the end
of all these transformations we have
\begin{align}
\label{eqn205}
s^{(k)}_{AB} = X_{A_1} \otimes X_{B_1}.
\end{align} 

STEP 2C - Keeping track of all the Clifford unitaries that have acted so far on
the two parts and combining them as $U_A, U_B$, we can write the new
stabilizer group as $(U_A U_B I_C) \SC (U_A U_B I_C)^\dag = \langle
t_1=X_{A_1} X_{B_1}, t_2=Z_{A_1}^{} Z_{B_1}^{-1}, t_3, \ldots \rangle$ by
Proposition \ref{proposition3}.
Since $t_3$ must commute with $t_1$ \emph{and} $t_2$, it must be of
the form
\begin{align}
\label{eqn206}
t_3 = X_{A_1}^{\alpha} Z_{A_1}^{\beta} p^{}_{A\backslash A_1} \otimes X_{B_1}^{\alpha} Z_{B_1}^{-\beta} q^{}_{B \backslash B_1} \otimes r^{}_C
\end{align}
for some $\alpha, \beta \in \ZZ_D$, where $p,q,r$ are some Pauli products on
subsystem $A$ less $A_1$, $B$ less $B_1$ and $C$ respectively, reusing
arguments that produced Eq.~\eqref{eqn204}. Next replace the generator $t_3$
by $t_3' = t_2^{-\beta} t_1^{-\alpha} t^{}_3$ so that $t_3'$ has identities on
qudits $A_1, B_1$. The same can be done for all other generators, $t_4,
\ldots, t_n$.

STEP 2D - Given this, we can write the new stabilizer group as $\langle
X_{A_1} X_{B_1}, Z_{A_1}^{} Z_{B_1}^{-1} \rangle \otimes \langle t_3', \ldots,
t_n' \rangle$: an EPR pair stabilized by $\langle X_{A_1} X_{B_1}, Z_{A_1}^{}
Z_{B_1}^{-1} \rangle$ and a stabilizer state $\ket{\SC'}$ stabilized by
$\langle t_3', \ldots, t_n' \rangle$.
The reduced density operators $\rho'_A$ and $\rho'_B$ corresponding to
$\ket{\SC'}$ are again proportional to the identity operators on the parts of
$A$ and $B$ that remain after the extraction.  Were it otherwise in the case
of $\rho'_A$, the reduced density operator on the (full) system $A$
corresponding to the state stabilized by $\langle X_{A_1} X_{B_1}, Z_{A_1}^{}
Z_{B_1}^{-1} \rangle \otimes \langle t_3', \ldots, t_n' \rangle$ would not be
proportional to the identity, contradicting the fact that after Step 1
(extraction of unentangled subsystems), and thus at the beginning of Step 2,
$\rho_A$ was proportional to the identity; of course the same applies to
$\rho'_B$. 

STEP 2E - Next examine the expansion of $\rho'_{AB}$ putting
$\ket{\SC'}$ in \eqref{eqn201}.  If some of the $s^{(i)}_{A}$ do not commute
with each other a further extraction is possible, and one can repeat the
process until all of these operators commute. Now apply the same extraction
process to $A$ and $C$, and then to $B$ and $C$, until all EPR pairs have been
extracted.

\bigskip

 The final step is showing that the tripartite stabilizer state after
 extracting all unentangled states and EPR pairs is a collection of GHZ states. 

STEP 3A - We prove that after the preceding extractions have been carried out, the three parts must have the same number of qudits, $n_A = n_B = n_C$. Since $\rho_C$ is proportional to the identity, $\ket{\SC}$ is maximally entangled across the $C$-$AB$ cut. Using its Schmidt form allows the the projector on $\ket{\SC}$ to be written as
\begin{align}
\label{eqn250}
\ket{\SC} \bra{\SC} = \frac{1}{D^{n_C}} \sum_{i,j=0}^{D^{n_C}-1} \ket{i} \bra{j}_C \otimes \ket{\phi_i} \bra{\phi_j}_{AB},
\end{align}
with $\{\ket{\phi_i} \}$ a set of orthonormal kets in $\HC_A \otimes \HC_B$.
The $\{\ket{i} \bra{j}_C\}$ are a set of $D^{2 n_C}$ linearly independent
operators, as are the $\{\ket{\phi_i} \bra{\phi_j}_{AB}\}$, which means the
\emph{operator Schmidt rank} of $\ket{\SC} \bra{\SC}$ is $D^{2 n_C}$.

From \Eq{eqn11} the projector can also be expressed as a sum of the $D^n$ linearly independent stabilizer elements in $\SC$
\begin{align}
\label{eqn251}
\ket{\SC} \bra{\SC} &= \frac{1}{D^{n}} \sum_{k=1}^{D^{n}} r^{(k)} \nonumber\\
&= \frac{1}{D^{n}} \sum_{k=1}^{D^{n}} r_C^{(k)}\otimes r_{AB}^{(k)}
\end{align}
where the $r_{AB}^{(k)}$ and $r_C^{(k)}$ are Pauli products on $AB$ and $C$, respectively. In general, the collection $\{r_C^{(k)}\}$ is not linearly independent; e.g., each $r^{(k)}$ that belongs to the subgroup $\SC_{AB}$ satisfies $r_C^{(k)} \propto I_C$. Indeed, two elements $r^{(k)}$ and $r^{(l)}$ belong to the same coset of $\SC_{AB}$ if and only if $r_C^{(k)} \propto r_C^{(l)}$. Therefore the number of cosets of $\SC_{AB}$ is the number of linearly independent elements in the collection $\{r_C^{(k)}\}_{k=1}^{D^n}$. We shall now prove that this number is $D^{2 n_C}$.

We can re-express \Eq{eqn251} as a sum over linearly independent Pauli products on part $C$
\begin{align}
\label{eqn252}
\ket{\SC} \bra{\SC} \propto \sum_{k'=1}^{D^{2 n_C}} q_C^{(k')} \otimes x_{AB}^{(k')},
\end{align}
where $\{x_{AB}^{(k')}\}$ are sums of Pauli products on parts $A$ and $B$. We know there must be $D^{2 n_C}$ linearly independent terms since that is the operator Schmidt rank and thus none of the $x_{AB}^{(k')}$ can vanish. Hence the number of cosets of $\SC_{AB}$ is $D^{2 n_C}$. Thus by Lagrange's theorem, 
\begin{align}
\label{eqn253}
|\SC_{AB}| = |\SC|/D^{2n_C} = D^{n_A+n_B-n_C}.
\end{align}

We now go on a digression to show that the collection of Pauli products on
subsystem $A$, $\{s_A^{(i)}\}$, on the right side of \Eq{eqn201}, is linearly
independent when $\rho_B \propto I_B$; similarly, the collection
$\{s_B^{(i)}\}$ is linearly independent if $\rho_A \propto I_A$.  Because they
are Pauli products it suffices to show that no two of them are proportional in
order to demonstrate linear independence.  Assume the contrary, that
$s^{(j)}_{A} \propto s^{(k)}_{A}$ for some $j \neq k$. Since the
$\{s_{AB}^{(i)}\}$ are group elements, $s_{AB}^{(j)}$ must have an inverse,
 $\left( s_{AB}^{(j)} \right)^{-1}$,
which inserted in \Eq{eqn201} yields
\begin{align}
\label{eqn254}
\left( s_{AB}^{(j)} \right)^{-1} s_{AB}^{(k)} &= \left( s_A^{(j)} \otimes s_B^{(j)} \right)^{-1} (s_A^{(k)} \otimes s_B^{(k)}) \nonumber \\
& \propto I_A \otimes \left(s_B^{(j)} \right)^{-1} s_B^{(k)}.
\end{align}
The final term cannot be proportional to the identity, as that would imply
$s_{AB}^{(j)} \propto s_{AB}^{(k)}$. But these are elements of $\SC$, so they
must be linearly independent for $j \neq k$. Therefore $\SC_B$
contains an element that is not the identity, contradicting the fact that
$\rho_B\propto I_B$. Hence it cannot be the case that $s^{(j)}_{A} \propto
s^{(k)}_{A}$ for $j\neq k$.  

Thus the collection $\{s_A^{(i)}\}$ contains $|\SC_{AB}|$ linearly independent\emph{and} mutually commuting elements, and by Proposition \ref{proposition1} this means that $|\SC_{AB}| \leqslant D^{n_A}$. The same argument applies to the collection $\{s_B^{(i)}\}$, so $|\SC_{AB}| \leqslant D^{n_B}$. Combining these inequalities with \Eq{eqn253}, it follows that both $n_A$ and $n_B$ cannot be larger than $n_C$. Identical arguments applied to different pairs of subsystems implies that
\begin{align}
\label{eqn255}
n_A&=n_B=n_C,
\end{align}
and, using  \Eq{eqn253}, $|\SC_{AB}| = |\SC_{BC}| = |\SC_{AC}|= D^{n_A}$.

STEP 3B - If $\SC_{BC}$ contains only the
identity element, then $n_A=n_B=n_C=0$ and no GHZ state can be
extracted. Otherwise $n_A \geqslant 1$ and $\SC_{BC}$ has at least one
non-trivial element which we will label as $t_1$. By Proposition
\ref{proposition2} it can be transformed to
\begin{align}
\label{eqn256}
t_1 = Z_{B_1}^{} \otimes Z_{C_1}^{-1}.
\end{align}

Now consider the group $\SC_{AB}$. We showed previously that the operators
$\{s_B^{(i)}\}$ defined in Eq.~\eqref{eqn201} are both linearly independent
and mutually commuting, and $|\{s_B^{(i)}\}| = D^{n_B}$. Hence, if there is an
element $p \in \PC_{n_B}$ that commutes with every element in $\{s_B^{(i)}\}$,
it must belong to $\{s_B^{(i)}\}$ up to a phase. This is because, by
Proposition \ref{proposition1}, sets of mutually commuting Pauli products on
$n_B$ qudits cannot have more than $D^{n_B}$ linearly independent elements.

Note that $t_1$ commuting with every element of $\SC_{AB}$ (as they are all just elements of $\SC$) implies $ Z_{B_1}^{}$ commutes with every $s_B^{(i)}$ since $t_1$ is identity on part $A$ and every $s \in \SC_{AB}$ is identity on part $C$. Then by the argument in the previous paragraph, there exists an element in $\SC_{AB}$ of the form $p_A \otimes Z_{B_1}$. Let $t_2$ denote this element and note that $p_A \neq I_A$ as otherwise this would contradict the assumption that $\rho_B$ is proportional to the identity. Then by Proposition \ref{proposition2}, there exists unitary transformations such that
\begin{align}
\label{eqn257}
t_2 = Z_{A_1}^{-1} \otimes Z_{B_1}^{}
\end{align}
and they do not affect $t_1$ since all the operations are done on part $A$.

STEP 3C - Since each of the $D^{2 n_C}$ Pauli products for part $C$ must appear
in Eq.~\eqref{eqn252}---none of $x_{AB}^{(k')}$ can vanish---there exists an
element in $\SC$ which satisfies $q_{C}^{(k')} \propto X_{C_1} I_{C \backslash C_1}$ on the $C$ subsystem. We call that element $t_3$. The most general form it can have,
given that it has to commute with $t_1 = Z_{B_1}^{} Z_{C_1}^{-1}$ and $t_2 =
Z_{A_1}^{-1} Z_{B_1}^{}$, is
\begin{align}
\label{eqn258}
t_3 = X_{A_1} Z_{A_1}^{\alpha} p_{A \backslash A_1} \otimes X_{B_1} Z_{B_1}^{\beta} q_{B \backslash B_1} \otimes X_{C_1},
\end{align}
where $p,q$ are Pauli products on the remaining qudits in part $A,B$
respectively and $\alpha, \beta \in \ZZ_D$; see STEP 2B for the
explanation. Finally we transform this element to $t_3 = X_{A_1} X_{B_1}
X_{C_1}$ without modifying $t_1, t_2$ by using the techniques described in
STEP 2B and in the proof of Proposition \ref{proposition2}.

STEP 3D - Invoking Proposition \ref{proposition3} and letting $U_A, U_B, U_C$
denote all the unitary operations we have made so far, we can write $(U_A U_B
U_C) \SC (U_A U_B U_C)^\dag = \langle t_1 = Z_{B_1}^{} Z_{C_1}^{-1}, t_2 =
Z_{A_1}^{-1} Z_{B_1}^{}, t_3 = X_{A_1} X_{B_1} X_{C_1}, t_4, \ldots,
t_n\rangle$.
As was done in STEP 2C and the proof of Bipartition Theorem, the generators
$t_4, \ldots, t_n$ can all be made to be the identity on qudits $A_1, B_1,
C_1$ simultaneously, so we can write $U_A U_B U_C \SC (U_A U_B U_C)^\dag =
\langle Z_{B_1}^{} Z_{C_1}^{-1}, Z_{A_1}^{-1} Z_{B_1}^{}, X_{A_1} X_{B_1}
X_{C_1}\rangle \otimes \langle t'_4, \ldots, t'_n \rangle$. Therefore we have
a GHZ state, see Eq.~\eqref{eqn13}, on qudits $A_1, B_1, C_1$, tensored with a
state on a system which has one fewer qudit in each part. This extraction
process can be repeated until $n_A = n_B = n_C =0$.
\end{proof}

\xa
\outl{Some comments about tripartition theorem.}
\xb

%===============================================================================

\section{Application in Quantum Error Correction}
\label{section5}

\xa
\outl{Give some motivation about the problem.}
\xb

In this section we apply the Tripartition Theorem to solve a problem in the
area of quantum error correction. It allows us to understand the structure of
a class of quantum channels derived from qudit stabilizer codes (see
Subsection \ref{subsection2c}) of prime $D$ by decomposing them to a tensor
products of perfect quantum channels, perfectly decohering channels and
completely depolarizing channels.

The connection between tripartite stabilizer states and stabilizer codes is worked out in Subsection \ref{subsection5a} using Choi-Jamio\l kowski isomorphism or map-state duality. This leads to a class of channels which we call \emph{stabilizer code channels} defined in Subsection \ref{subsection5b}, whose decomposition into simple channels is the topic of Subsection \ref{subsection5c}. A duality between the subset information groups of such a channel and its complementary channel is demonstrated in Subsection \ref{subsection5d}.

\subsection{Isomorphism Between Stabilizer Codes and Stabilizer States}
\label{subsection5a}

Let $\HC_A$ and $\HC_B$ be two Hilbert spaces, and $\{\ket{i}_A\}$ an orthonormal basis for $\HC_A$. Then there is a one-to-one correspondence between a linear map $M$ from $\HC_A$ and $\HC_B$ and a ket $\ket{M}$ on the tensor product $\HC_A\otimes\HC_B$ conveniently expressed in Dirac notation as changing bras to kets or vice versa:
\begin{equation}
  M = \sum_i \ket{\beta_i}\bra{i}_A,\quad 
\ket{M} = \sum_i \ket{\beta_i}\otimes\ket{i}_A.
\label{eqn271}
\end{equation}
Here the $\{\ket{\beta_i}\}$ are elements of $\HC_B$, in general neither
orthogonal nor normalized, uniquely determined by $M$ or by $\ket{M}$ as the
case may be. This \emph{Choi-Jamio\l kowski isomorphism} %
\footnote{The idea goes back to Choi \cite{laa-10.285} and Jamio\l kowski \cite{jamiol} and even earlier; see \cite{quant-ph.0401119} for extensive references to the literature. The isomorphism is usually defined between quantum channels and bipartite density operators so technically our definition is just ``half" of the standard form, see \cite{PhysRevA.73.052309, arrighi-2004-311}. } %
depends on the choice of the basis $\{\ket{i}_A\}$; in what follows this will always be the computational basis. For convenience we introduce a normalization factor of $\sqrt{d_A}$ in defining the following isomorphism and its inverse:
\begin{equation}
  \Phi(M) = \ket{M}/\sqrt{d_A},\quad
 \Phi^{-1}(\ket{M}) = \sqrt{d_A}\, M,
\label{eqn272}
\end{equation}
with $M$ and $\ket{M}$ related by \Eq{eqn271}. It is straightforward to show that
\begin{align}
\Phi(U_B M U_A) &= (U_A^{T} U_B)\Phi(M), \nonumber \\ 
\Phi^{-1}\left[(U_B\otimes U_A)\ket{M}\right] &= U_B \Phi^{-1}(\ket{M})U_A^{T},
\label{eqn273}
\end{align}
where the transpose $T$ refers to the $\{\ket{i}_A\}$ basis, and $U_A$ and $U_B$ are any operators on $\HC_A$ and $\HC_B$, respectively.

Next assume that $d_A \leqslant d_B$ and define an isometry $V : \HC_A
\rightarrow \HC_B$ as
\begin{align}
\label{eqn300}
V &:=  \sum_{i=0}^{d_A} \ket{v_i}_B\bra{i}_A,
\end{align}
where $\{\ket{v_i}\}$ is an orthonormal basis spanning a $d_A$-dimensional
subspace of $\HC_B$ that we call the \emph{coding space}. The isomorphism
carries it into 
\begin{align}
\label{eqn301}
\ket{V} := \Phi(V) = 
\frac{1}{\sqrt{d_A}} \sum_{i=0}^{d_A} \ket{v_i}_B \ket{i}_A,
\end{align}
where, since $V$ is an isometry, $\rho_A = \Tr_{B} \left\{\ket{V}
  \bra{V}\right\} = I_A/d_A$. One can think of $\ket{V}$, where $\HC_A$ and
$\HC_B$ enter on an equal footing, as an \emph{atemporal} representation of
the isometry \cite{PhysRevA.73.052309}.

The following lemma uses the Choi-Jamio\l kowski isomorphism to relate isometries corresponding to stabilizer codes to stabilizer states. 

\begin{lemma}[Isomorphism between stabilizer code and bipartite stabilizer
  state]
\label{lemma8}
Let $D$ be the dimension of any one of the qudits. The two statements below
are true for any prime $D$
\begin{enumerate}
\item Given an $[[n, k]]_D$ stabilizer code (Subsection
  \ref{subsection2c}) that defines an isometry $V : \HC^{\otimes k}
  \rightarrow \HC^{\otimes n}$, the isomorphism $\Phi$ in 
  \Eq{eqn272} carries $V$ to a stabilizer state $\ket{V}$ on $k+n$
  qudits.

\item
Let $\ket{\SC}_{AB}$ be a stabilizer state on $k+n$ qudits,
 where the first $k$ qudits constitute part $A$
and the remaining $n$ qudits constitute part $B$, and assume that
$\Tr_{B} \left\{\ket{\SC} \bra{\SC}\right\} = I_A/D^k$.
Then the inverse isomorphism $\Phi^{-1}$ in \Eq{eqn272} carries
$\ket{\SC}_{AB}$ to an isometry $V : \HC_A \rightarrow \HC_B$ whose image or
coding space is an $[[n, k]]_D$ stabilizer code.
\end{enumerate}
\end{lemma}

\begin{proof}
  The proof can be simplified by noting that, as shown in
  \cite{quantph.0703112, quantph.0111080}, any stabilizer code when $D$ is
  prime is equivalent to an additive graph code, up to products of
  single-qudit Clifford unitaries. In turn it was shown in
  \cite{PhysRevA.78.042303} and \cite{PhysRevA.81.032326} that an additive
  $[[n, k]]_D$ graph code when $D$ is prime can be described using a graph
  state $\ket{G} \in \HC^{\otimes n}$ and an abelian group $\FC =
  \{c_i\}_{i=0}^{D^k-1}$ with $D^k$ linearly independent Pauli products in
  $\PC_n$, composed only of $Z$ operators. The coding space is spanned by the set $\{c_i \ket{G}\}$ of mutually orthogonal
  kets.
Recall from Subsection \ref{subsection2c} that any graph state $\ket{G}$ on
$n$ qudits can be fully specified by its stabilizer group, $\langle g_1,
\ldots, g_n\rangle$ with $D^n$ elements, see \Eq{eqn53}. Similarly, the coding
group $\FC$ can be generated by $k$ suitably chosen group elements, $\FC =
\langle f_1, f_2, \ldots, f_k\rangle$, so the coding space is spanned by
$\{f_{1}^{i_1} \cdots f_{k}^{i_k} \ket{G}\}$ for $i_1, i_2, \ldots, i_k =
0,1,\ldots,D-1$.

To prove statement 1, define the isometry 
\begin{align}
\label{eqn302}
V &= \sum_{i_1=0}^{D-1} \cdots \sum_{i_k=0}^{D-1} f_{1}^{i_1}  
\cdots f_{k}^{i_k} \ket{G}_B \bra{i_1 \cdots i_k}_A
\end{align}
that maps $\HC_A = \HC^{\otimes k}$, a collection of $k$ qudits, into the coding space,
and 
\begin{align}
\label{eqn303}
\ket{V} &= \frac{1}{\sqrt{D^k}}\sum_{i_1=0}^{D-1} \cdots 
\sum_{i_k=0}^{D-1}  \ket{i_1 \cdots i_k}_A \otimes f_{1}^{i_1}\cdots f_{k}^{i_k} \ket{G}_B.
\end{align}
the corresponding ket on $\HC_A\otimes \HC_B$ as per \Eq{eqn272}.

We shall show that $\ket{V}$ is a stabilizer state by exhibiting the $k+n$
stabilizer generators. The first $n$ are derived from the generators
$\{g_j\}$, for $j=1,2,\ldots, n$ of the stabilizer group of $\ket{G}$, now
regarded as operators on $\HC_A\otimes \HC_B$, so that
\begin{align}
\label{eqn304}
(I_A \otimes g_j) \ket{V} &= \frac{1}{\sqrt{D^k}}\sum_{i_1=0}^{D-1} \cdots
\sum_{i_k=0}^{D-1}  \ket{i_1 \cdots i_k}_A 
\nonumber\\
& \quad \otimes g_j \; f_{1}^{i_1} \cdots f_{k}^{i_k} \ket{G}_B 
\nonumber \\
&= \frac{1}{\sqrt{D^k}}\sum_{i_1=0}^{D-1} \cdots \sum_{i_k=0}^{D-1}  
\ket{i_1  \cdots i_k}_A 
\nonumber \\
& \quad \otimes \omega^{i_1 \beta_{j1}} \cdots 
\omega^{i_k \beta_{jk}} f_{1}^{i_1} \cdots f_{k}^{i_k} \ket{G}_B,
\end{align}
where the phases result from commuting the $g_j$ with the $f_l$ operators:
\begin{align}
\label{eqn305}
g_j f_l &= \omega^{\beta_{jl}} f_l g_j 
\quad \text{for $j=1,\ldots, n$ and $l = 1,\ldots, k$};
\end{align}
with the $\beta_{jl}$ integers in $\ZZ_D$.  The phases can be removed using
an appropriate Pauli product of $Z$ operators:
\begin{align}
\label{eqn306}
\left( Z_{A_1}^{-\beta_{j1}} \cdots 
Z_{A_k}^{-\beta_{jk}}\otimes g_j \right)\ket{ V} 
&= \ket{V} \quad \text{for $j=1,\ldots, n$},
\end{align}
That these $n$ generators mutually commute follows from the fact that 
$\{g_j\}$ is a mutually commuting set.

The remaining $k$ stabilizer generators are simply
\begin{equation}
\label{eqn307}
X_{A_l} \otimes (f_l^{-1})_B  \quad \text{for $l = 1,\ldots, k$.}
\end{equation}
They obviously commute among themselves.  It is not hard to show that they
leave $\ket{V}$ in Eq.~\eqref{eqn303} invariant, and commute with the $n$
previous generators on the left side of \Eq{eqn306}.  Since we have constructed
$k+n$ linearly independent commuting Pauli operators that leave $\ket{V}$
invariant, and thus generate a stabilizer group, $\ket{V}$ is a stabilizer
state.

\bigskip

To prove statement 2 of the lemma we use the fact that for prime $D$, any
stabilizer state $\ket{\SC}$ is equivalent up to single-qudit Clifford
unitaries, thus a choice of basis for the individual qudits, to a graph
state. Hence without loss of generality we can assume that $\ket{\SC}$ is a
graph state on $\HC_A\otimes\HC_B$ written in the form, see Eq.~\eqref{eqn51},
\begin{align}
\label{eqn309}
\ket{\SC}_{AB} &= \left( \prod_{i<j}^{n+k} \CP_{ij}^{\Gamma_{ij}} \right) \ket{+}_A^{\otimes k} \otimes \ket{+}_B^{\otimes n}.
%U_A U_B \ket{\SC}_{AB} &= \left( \prod_{i=1}^{n+k-1} \prod_{j=i+1}^{n+k} \CP_{ij}^{\Gamma_{ij}} \right) \ket{+}_A^{\otimes k} \otimes \ket{+}_B^{\otimes n}.
\end{align}
Let $\CP_A$ be the product of those $\CP$ gates on the right side that act only
on qudits in $A$, and let $\ket{G}_B = \CP_B \ket{+}_B^{\otimes n}$ be the graph state on $\HC_B$ resulting from the $\CP$ gates that act only on qudits in $B$, denoted by $\CP_B$. The action of the remaining $\CP$ gates that connect $A$ and $B$ qudits can be written out explicitly in terms of $Z$ operators, see \Eq{eqn8}, to obtain
\begin{align}
\label{eqn310}
\ket{\SC} &= \CP_A \left( \prod_{i=1}^{k} \prod_{j=k+1}^{n+k}
  \CP_{ij}^{\Gamma_{ij}} \right)   \ket{+}_A^{\otimes k} \ket{G}_B 
\nonumber \\ 
%&= \frac{1}{\sqrt{D^k}} \CP_A \sum_{i_1=0}^{D-1} \cdots \sum_{i_k=0}^{D-1}
%\ket{i_1 \cdots i_k}_A 
%\nonumber \\
%& \quad \otimes \left( \prod_{j=k+1}^{n+k} Z_j^{i_1 \Gamma_{1j}}  \cdots
%  Z_j^{i_k \Gamma_{kj}}  \right) \ket{G}_B 
%\nonumber \\
&= \frac{1}{\sqrt{D^k}} \CP_A \sum_{i_1=0}^{D-1} \cdots \sum_{i_k=0}^{D-1}
\ket{i_1 \cdots i_k}_A 
\nonumber \\
& \quad \otimes f_1^{i_1} \cdots f_k^{i_k} \ket{G}_B
\end{align}
where
\begin{align}
\label{eqn311}
f_l &:= \prod_{j=k+1}^{n+k} Z_j^{\Gamma_{lj}} \qquad \text{for $l=1,2,\ldots,k$,}
\end{align}
and so
\begin{align}
\prod_{j=k+1}^{n+k} \CP_{lj}^{\Gamma_{lj}} &= \sum_{i_l=0}^{D-1} \ket{i_l} \bra{i_l} \otimes f_l^{i_l} \qquad \text{for $l = 1,\ldots, k$.}
\end{align}

Comparing the last line of Eq.~\eqref{eqn310} with Eq.~\eqref{eqn303} yields
\begin{align}
\label{eqn312}
\Phi^{-1} [(\CP_A)^\dag \ket{\SC}] &= \sum_{i_1=0}^{D-1} \cdots
\sum_{i_k=0}^{D-1} f_1^{i_1} \cdots f_k^{i_k} \ket{G}_B 
\nonumber \\
&\quad \bra{i_1 \cdots i_k}_A,
\end{align}
which is an isometry corresponding to an additive graph code with graph state
$\ket{G}_B \in \HC^{\otimes n}$ and coding group $\CC = \langle f_1,
\ldots, f_k\rangle$. Therefore, using \Eq{eqn273}
\begin{align}
\label{eqn313}
\Phi^{-1} (\ket{\SC}) &= \sum_{i_1=0}^{D-1} \cdots \sum_{i_k=0}^{D-1} U_B
f_1^{i_1} \cdots f_k^{i_k} \ket{G}_B 
\nonumber \\
& \quad \bra{i_1 \cdots i_k}_A [(\CP_A)^\dag ]^T,
\end{align}
where the transpose is taken in the computational basis. The final factor is
simply a unitary transformation on the input and thus does not change its
image, which is the graph or stabilizer code spanned by the $f_1^{i_1} \cdots
f_k^{i_k} \ket{G}_B$.  That these last are a collection of $D^k$ mutually
orthogonal kets follows from the assumption that $\Tr_{B} \left\{ \ket{\SC}
  \bra{\SC} \right\} = I_A/D^k$, and the fact that the final equality in
\Eq{eqn310} is a Schmidt decomposition of $\ket{\SC}$.
\end{proof}

\subsection{Stabilizer Code Channels and Subset Information Groups}
\label{subsection5b}

\xa
\outl{Define SCE channel.}
\xb

Consider an isometry $V : \HC_A=\HC^{\otimes k} \rightarrow \HC^{\otimes n}$
corresponding to an $[[n,k]]_D$ stabilizer code where $D$ is prime and the $n$
output qudits are partitioned into two disjoint non-empty subsets, $B$ and
$C$. Such bipartitions of stabilizer/graph codes have been studied in
\cite{PhysRevA.81.032326} and \cite{PhysRevA.81.052302}. (Qubit stabilizer
codes with input qudits partitioned into two parts have also been
studied in \cite{springerlink:10.1007/s11128-010-0175-0} but in this paper we
will only consider bipartitions of the output qudits.)
We can think of the $B$ qudits as the output of a \emph{direct} quantum
channel, and the $C$ qudits either as the environment or as the output of
\emph{complementary} channel, with corresponding superoperators
\begin{align}
\label{eqn350}
\EC_B(\rho) := \Tr_C\{ V \rho V^\dag\} \qquad
 \EC_C(\rho) := \Tr_B\{ V \rho V^\dag\}.
\end{align}
We shall refer to channels derived in this way from stabilizer codes as
\emph{stabilizer code channels}. The analysis below applies to any stabilizer
code regardless of its error correction properties such as code distance.

\xa
\outl{Mention results on subset information group from previous paper.}
\xb

In Section V of \cite{PhysRevA.81.032326}, we studied the information-carrying
capacity of stabilizer code channels and to this end introduced the 
\emph{subset
  information group}
\begin{align}
\label{eqn351}
\GC_B := \{p \in \PC_k \;|\; \EC_B(p) \neq 0 \},
\end{align}
a subgroup of the Pauli group on the $k$ input qudits.
It was shown that $\GC_B$ is a group, and its elements satisfy the isomorphism
\begin{align}
\label{eqn352}
\EC_B(p) \; \EC_B(q) = c \; \EC_B(pq ) \quad \forall \; p,q \in \GC_B,
\end{align}
where $c$ is an appropriately chosen positive constant independent of $p$ and $q$.

We also presented an efficient algorithm to find $\GC_B$ given the isometry
$V$ (defined by an additive graph code) and the subset $B$, by solving a set
of linear equations modulo $D$. The subset information group $\GC_C$ for the
complementary channel can be defined in the same way. The previous work
discussed additive graph codes, but the results apply to stabilizer codes as
well since they are equivalent up to local unitaries.

\xa
\outl{Connection to Beny's work.}
\xb

One can think of the group $\GC_B$ or the operator algebra that it spans, a
subalgebra of the algebra $\LC(\HC_A)$ of operators on the channel input, as
representing the information that is perfectly transmitted from the input to
the output $\HC_B$ of the channel by $\EC_B$.  Since it is present in the
output, this information can be perfectly recovered, which is to say mapped to
a Hilbert space $\HC'_A$ isomorphic to $\HC_A$, by a recovery operation
(recovery channel) $\RC: \LC(\HC_B)\rightarrow \LC(\HC'_A)$, as shown in
\cite{PhysRevA.81.032326}. The recovery operation has the property that $[\RC
\circ \EC_B]^\dag (x) = x$, where $x$ is any Pauli product in $\GC_B$ or any
operator in the subalgebra that it spans.  Thus the last is an example of a
\emph{correctable algebra}, as defined in \cite{PhysRevA.76.042303} in their formalism of operator algebra quantum error correction.

In the following subsections we show that stabilizer code channels can
be decomposed into tensor products of simple channels, which has
important implications for the properties of $\GC_B$ and $\GC_C$.

\subsection{Tensor Product Structure of Stabilizer Code Channels}
\label{subsection5c}
\xa
\outl{Briefly summarize results to be presented.}
\xb

Before stating the main result in Theorem~\ref{theorem9}, let us indicate by
means of some simple examples its main idea, which is that the isomorphism proven in Lemma \ref{lemma8} with Theorem~\ref{theorem7} imply stabilizer code channels have a
very simple structure.

\xa
\outl{Consider a few simple states and work out the corresponding isometries, subset information groups and channel capacities.}
\xb

First consider the stabilizer state, $\ket{V} = \ket{\mathrm{EPR}}_{A_1 B_1}$,
Eq.~\eqref{eqn12}, which by Eqs.~\eqref{eqn300} and
\eqref{eqn301} corresponds to the isometry
\begin{align}
\label{eqn353}
\Phi^{-1}(\ket{V}) = V_{\mathrm{EPR}} &= \sum_{i=0}^{D-1} \ket{i}_{B_1} \bra{i}_{A_1} 
\end{align}
for a perfect quantum channel from $\HC_{A_1}$ to $\HC_{B_1}$, with quantum
(and classical) channel capacity equal to $\log_2 D$. The subset
information group contains \emph{every} Pauli operator on qudit $A_1$,
i.e. $\GC_B = \langle \lambda I_{A_1}, X_{A_1}, Z_{A_1} \rangle$.

Next consider the tripartite state $\ket{V} = \ket{\mathrm{EPR}}_{A_1
  B_1}\otimes \ket{\mathrm{EPR}}_{B_2 C_1}$. Tracing out part $C$,
Eq.\eqref{eqn350}, yields the channel $\EC_B$ with the same $\GC_B = \langle
\lambda I_{A_1}, X_{A_1}, Z_{A_1} \rangle$ as before, while the complementary
channel $\EC_C$ is the completely noisy or completely depolarizing channel
whose subset information group is (multiples of) the identity, $\GC_C =
\langle \lambda I_{A_1} \rangle$. Therefore, attaching the state
$\ket{\mathrm{EPR}}_{B_2 C_1}$ to that of the previous example increases the
dimension of the output Hilbert space while leaving the $\HC_A$ to
$\HC_B$ channel unchanged. This is not surprising, since the
$\ket{\mathrm{EPR}}_{B_2 C_1}$ part has nothing to do with the input Hilbert
space $\HC_A$.

\xa
\outl{Isometry obtained from GHZ.}
\xb

As a third example, the tripartite $\ket{\mathrm{GHZ}}$ state from \Eq{eqn13}
\begin{align}
\label{eqn354}
\ket{V} &= \ket{\mathrm{GHZ}}_{A_1 B_1 C_1} =
\frac{1}{\sqrt{D}} \sum_{i=0}^{D-1} \ket{i}_{A_1} \ket{i}_{B_1} \ket{i}_{C_1}
\end{align} 
is carried by the inverse map in \Eq{eqn272} to the isometry
\begin{align}
\label{eqn355}
\Phi^{-1}(\ket{V}) = V_{\mathrm{GHZ}} &= \sum_{i=0}^{D-1} \ket{i}_{B_1} \ket{i}_{C_1} \bra{i}_{A_1}.
\end{align} 
In this case, $\EC_B$ and $\EC_C$ are \emph{perfectly decohering
channels} whose Kraus representation is
\begin{align}
\label{eqn356}
\EC_B(\rho) = \EC_C(\rho) = \sum_{i=0}^{D-1} \ket{i} \bra{i}
\rho \ket{i} \bra{i},
\end{align}
generalizing the qubit phase-flip channel (see Chap. 8 of \cite{QuantumComputation}) to arbitrary $D$. The quantum
channel capacity is zero while the classical channel capacity is $\log_2
D$. The subset information groups for both channels are identical: $\GC_B =
\GC_C = \langle \lambda I_{A_1}, Z_{A_1} \rangle$.

The following theorem states that the isometry of any stabilizer code with a
bipartition defined on the output qudits is equivalent to a tensor product of
isometries of the form $V_{\mathrm{EPR}}$ and $V_{\mathrm{GHZ}}$. See
Figure~\ref{isometry_decomp} for an example of a decomposition of an isometry
of a $[[7,3]]_D$ stabilizer code with $n_B=3$ and $n_C=4$.

\begin{theorem}[Tensor product structure of stabilizer code isometries]
\label{theorem9}
For prime $D$, let $V:\HC_A \rightarrow \HC_B \otimes \HC_C$ be an isometry
corresponding to an $[[n,k]]_D$ stabilizer code with a $B$-$C$ bipartition of
the $n$ output qudits. Then up to unitaries $U_A, U_B, U_C$ on $\HC_A$,
$\HC_B$ and $\HC_C$, $V$ is a tensor product,
\begin{align}
\label{eqn357}
U_B U_C V U_A &= V_{\mathrm{GHZ}}^{\otimes m_{ABC}} \otimes
V_{\mathrm{EPR}}^{\otimes m_{AB}} \otimes V_{\mathrm{EPR}}^{\otimes m_{AC}} 
\nonumber \\
&\quad \otimes \ket{\mathrm{EPR}}^{\otimes m_{BC}} 
\otimes \ket{+}^{\otimes m_B} \otimes \ket{+}^{\otimes m_C},
\end{align}
where $m_{ABC},\, m_{AB},\, m_{AC},\, m_{BC},\, m_B,\, m_C$ are non-negative integers, and $V_{\mathrm{EPR}}$ and $V_{\mathrm{GHZ}}$ are defined in Eqs.~\eqref{eqn353} and \eqref{eqn355}.
\end{theorem}

\begin{proof}
  Use Lemma \ref{lemma8} to map the stabilizer code isometry $V$ to a
  stabilizer state $ \Phi(V) = \ket{V}$ on $k+n$ qudits.  Use
  Theorem~\ref{theorem7} to express the result, up to local unitaries, in the
  form given in \Eq{eqn200}.  Apply to this the inverse map $\Phi^{-1}$, noting    that unitaries can be pulled outside, as shown in \Eq{eqn273}.
\end{proof}

\begin{figure}%[ht]
\begin{center}
\includegraphics{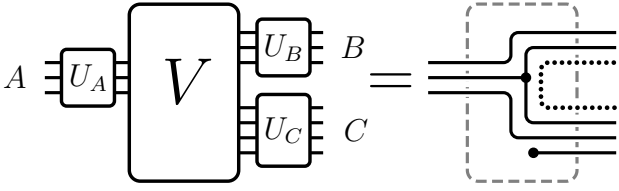}
\caption{Decomposition of the isometry of a $[[7,3]]_D$ stabilizer code into simple isometries as proven in Theorem \ref{theorem9}. The dotted line in the diagram on the right represents an EPR pair between parts $B$ and $C$.}
\label{isometry_decomp}
\end{center}
\end{figure}

An immediate corollary is that stabilizer code channels, up to
unitaries on the input and output spaces, are always tensor products of just
three types of simple channel; (i) Perfect quantum channel, (ii) Perfectly
decohering or phase-flip channel, (iii) Completely noisy or completely
depolarizing channel. This determines the quantum and classical capacities
since, e.g., the quantum capacity of $\EC_B$ is just the number of EPR pairs
that can be extracted from the $A$-$B$ bipartition of $\ket{V}$ multiplied by
$\log_2 D$, while the classical capacity is the quantum capacity plus the
number of GHZ states linking $A$, $B$, and $C$.  

Hence it follows that if $\EC$ is a stabilizer code channel and its classical and quantum channel capacities are known, then one can deduce the decomposition of the isometry as in \Eq{eqn357} and Figure \ref{isometry_decomp}, since the dimension of the input is known. From here it is straightforward to work out the classical and quantum channel capacities of the \emph{complementary} channel.

\subsection{Duality of Subset Information Groups, $\GC_B$ and $\GC_C$}
\label{subsection5d}

Here we shall demonstrate how the tensor product structure of stabilizer code
channels enables us to easily determine the subset information groups, $\GC_B$
and $\GC_C$. We also prove a duality relation between the $\GC_B$ and
$\GC_C$ which implies that each of them is determined by the other.

\xa
\outl{Definition of centralizer.}
\xb

Recall that the subset information groups $\GC_B$ and $\GC_C$ are subgroups of
the Pauli group $\PC_k$ on the $k$ input qudits, and that the
\emph{centralizer} $\mathrm{Cent}_{\GC} (\AC)$ of any subset $\AC$ of elements
of a group $\GC$ consists of all elements of $\GC$ that commute with every 
member of $\AC$. 

\begin{theorem}
\label{theorem10}
For $D$ prime, let $V: \HC_A \rightarrow \HC_B \otimes \HC_C$ be an isometry
corresponding to an $[[n, k]]_D$ stabilizer code, with $\PC_k$ the Pauli group
on the input qudits. Assume a $B$-$C$ bipartition of the $n$ output qudits is defined, with $\GC_B$ and $\GC_C$ the corresponding subset
information groups as defined in Eq.~\eqref{eqn351} and its counterpart with
$B$ replaced with $C$, for the stabilizer code channels $\EC_B$ and $\EC_C$.

Then $\GC_B$ and $\GC_C$ are dual in the sense that
\begin{align}
\label{eqn403}
\GC_B = \mathrm{Cent}_{\PC_k}(\GC_C) \quad \text{and} \quad 
\GC_C = \mathrm{Cent}_{\PC_k}(\GC_B),
\end{align}
i.e., each is the centralizer of the other in the group $\PC_k$.
\end{theorem}

\begin{proof}
  First we show that the duality relation holds for $V_{\mathrm{EPR}}$ and for
  $V_{\mathrm{GHZ}}$, as defined in Eqs.~\eqref{eqn353} and \eqref{eqn355}
  respectively. For $V_{\mathrm{EPR}}^{A_1 \rightarrow B_1}$, we have $\GC_B =
  \langle \lambda I_{A_1}, X_{A_1}, Z_{A_1} \rangle$ from the previous
  subsection. The complementary channel is $\EC_C(p) = \Tr_B\{VpV^\dag\} =
  \Tr\{p\}$, and by Eq.~\eqref{eqn351} we have $\GC_C = \langle \lambda
  I_{A_1} \rangle$. Thus $\GC_B$ is the full Pauli group, and $\GC_C$ consists
  of multiples of the identity, so they are centralizers of each other in the
  Pauli group on qudit $A_1$.  Next, for the isometry $V_{\mathrm{GHZ}}^{A_1
    \rightarrow B_1 C_1}$ we have $\GC_B = \GC_C = \langle \lambda I_{A_1},
  Z_{A_1} \rangle$ and again the duality is satisfied.
  
But Theorem~\ref{theorem9} says that $V$ is equivalent up to local unitaries to tensor products of these two types of elementary isometries, where we ignore the single-qudits states on $B$, $C$ and the EPR pairs between parts $B$ and $C$ as they are not involved in the transmission of information. It is straightforward to show that since the Pauli group for part $A$ and the information groups are themselves tensor products, the latter are again centralizers of each other.  Since the duality is invariant under conjugation by unitaries, the local unitaries play no role.
\end{proof}

Here is an illustrative example.  Suppose that
\begin{align}
\label{eqn359}
V = V_{\mathrm{EPR}}^{A_1 \rightarrow B_1} \otimes 
V_{\mathrm{EPR}}^{A_2 \rightarrow C_1} \otimes 
V_{\mathrm{GHZ}}^{A_3 \rightarrow B_2 C_2}.
\end{align}
Then the subset information groups have the same tensor products structure as the isometries above
\begin{align}
\label{eqn360}
\GC_B &= \langle \lambda I \rangle \otimes \langle X_{A_1}, Z_{A_1} \rangle \otimes \langle I_{A_2} \rangle  \otimes \langle Z_{A_3} \rangle \nonumber\\
&= \langle \lambda I, X_{A_1}, Z_{A_1}, Z_{A_3} \rangle, \nonumber \\
\GC_C &= \langle \lambda I \rangle \otimes \langle I_{A_1} \rangle  \otimes  \langle X_{A_2}, Z_{A_2} \rangle \otimes \langle Z_{A_3} \rangle \nonumber\\
&= \langle \lambda I, X_{A_2}, Z_{A_2}, Z_{A_3} \rangle.
\end{align}

\xa
\outl{Comments/implications of result.}
\xb

While the results presented in this section only hold for stabilizer code
isometries, we now show that they can sometimes be used to derive bounds for
more general isometries $V$, with channels $\EC_B$ and $\EC_C$ defined as in
\Eq{eqn350} for some bipartition of the output qudits. Given a general coding
space $\VC$, one can ask if there are stabilizer subspaces or subcodes
\emph{contained} in $\VC$. If these subcodes are not one-dimensional
subspaces, then meaningful lower bounds on channel capacities can be calculated for $\EC_B$ and $\EC_C$.

For example consider the nonadditive (or non-stabilizer) qubit graph code denoted by $((5,6,2))_2$, mentioned in \cite{PhysRevA.78.042303} and
\cite{quant-ph.0708.1021}. (It was first described in
\cite{PhysRevLett.79.953}, but not using the graph code formalism.) Also assume
the following bipartition of the five output qubits, $B = \{1, 2\}$ and
$C=\{3,4,5\}$. The six-dimensional coding space is spanned by
\begin{align}
\label{eqn361}
\VC &= \mathrm{span} \{ \ket{G}, Z_1 Z_2 Z_4 \ket{G}, Z_2 Z_3 Z_5 \ket{G}, 
\nonumber \\
& \quad Z_1 Z_3 Z_4 \ket{G}, Z_2 Z_4 Z_5 \ket{G}, Z_1 Z_3 Z_5 \ket{G} \},
\end{align}
where $\ket{G}$ is a five-qubit graph code stabilized by the group $\langle X_1 Z_2 Z_5, Z_1 X_2 Z_3, Z_2 X_3 Z_4, Z_3 X_4 Z_5, Z_1 Z_4 X_5 \rangle$ and the corresponding graph is a pentagon \cite{PhysRevA.78.042303}.

Obviously, $\VC$ contains the two-dimensional stabilizer subcode
$\VC_0 \subset \VC$,
\begin{align}
\label{eqn362}
\VC_0 &= \mathrm{span} \{ \ket{G}, Z_1 Z_2 Z_4 \ket{G} \}.
\end{align}
By applying the techniques used in proving Lemma~\ref{lemma8} and
Theorem~\ref{theorem9} one can show that $\VC_0$ corresponds to a one qubit
quantum channel from $A$ to $C$, which means that the quantum channel capacity
of $\EC_C$ is greater or equal to $\log_2(2)$. 
Similarly, by  considering the subcode
\begin{align}
\label{eqn363}
\VC_1 &= \mathrm{span} \{ \ket{G}, Z_2 Z_3 Z_5 \ket{G} \},
\end{align}
one can identify a GHZ state in the corresponding tripartition, and thus
deduce that the classical channel capacity of $\EC_B$ and of $\EC_C$ is at
least $\log_2(2)$.

Working from the other direction, one can also consider stabilizer codes that
\emph{contain} $\VC$ as a subcode. Provided these stabilizer codes are not the
whole Hilbert space, useful upper bounds on channel capacities of $\EC_B$ and
$\EC_C$ can be calculated.

%=============================================================================

\section{Conclusion}
\label{conclusion}

\xa
\outl{Summary.}
\xb

The most important results of our paper are those in Corollary~\ref{corollary5}, Theorem~\ref{theorem7}, and Theorem~\ref{theorem9}. The first of these allows stabilizer states for composite $D$ to be expressed as tensor products of stabilizer states associated with the different prime factors of $D$. This is a valuable technical tool used later in the paper, but also a helpful conceptual tool as it allows more complicated cases to be ``pulled apart'' into simpler situations. For example, a graph state for $D=6$ can be regarded as the tensor product of $D=2$ and $D=3$ graph states.

Our main result, Theorem~\ref{theorem7}, that a tripartite stabilizer state
can be considered the tensor product of single-qudit states, two-qudit EPR
pairs and three-qudit GHZ states, generalizes the $D=2$ (qubit) result in
\cite{quantph.0504208} to any squarefree $D>2$.  This allows us to provide a
very simple and essentially complete characterization, Theorem~\ref{theorem9},
up to unitaries on the input and output, of channels constructed from
stabilizer quantum codes.  Knowledge of the information carrying properties of
such a channel leads immediately tells one the properties of the complementary
channel, and there is a simple relationship between the corresponding subset
information groups introduced in \cite{PhysRevA.81.032326}. In some cases one can use these results to put bounds on capacities of other types of channel.

There are various directions in which one might hope to extend these results.
We have encountered technical difficulties in attempting to generalize our
tripartition theorem from squarefree to arbitrary composite $D$, where it may
no longer be true.  It follows from Corollary \ref{corollary5} that it is
sufficient to resolve the situation in which $D$ is a prime power, so if that
could be solved one could have results that apply for general $D$.   
Can any of our results be extended beyond the narrow confines of stabilizer
states?  The fact that the proofs depend heavily on group-theoretical
properties makes this seem unlikely, but it would certainly be of interest to
understand what it is that makes stabilizer states so special, and stabilizer
quantum codes so useful. 

Another possible direction is to move on from tripartitions to those involving four or more parts.  Here the results in \cite{quantph.0504208} for qubits suggest a situation that is distinctly more complicated than found for bipartitions and tripartitions, and $D>2$ is unlikely to be simpler. But it would still merit study.

% =============================================================================

\begin{acknowledgments}
The authors would like to thank Patrick Coles, Vlad Gheorghiu and Dan Stahlke for helpful comments on the manuscript. The research described here received support from the National Science Foundation through Grant No. PHY-0757251.
\end{acknowledgments}

%=============================================================================

\appendix

\section{Proof of Theorem \ref{theorem4}}
\begin{proof}
It is crucial to first note that there are \emph{two} different notions of tensor product used in the statement of the theorem and in this proof. The first and more obvious notion is the tensor product of $n$ qudits, each of dimension $D$. The second notion is defined on the space of each qudit of dimension $D$, which is isomorphic to a tensor product of $m$ spaces on qudits of dimensions $d_1, d_2, \ldots, d_m$.

First of all, we show that Pauli operators on a single qudit of dimension $D$ are equivalent to tensor products of $m$ Pauli products on the constituent qudits. We will be relying primarily on the Chinese Remainder Theorem which states that there exists a ring isomorphism between $\ZZ_D$ and the tensor product of rings, $\ZZ_{d_1} \otimes \ZZ_{d_2} \otimes \cdots \otimes \ZZ_{d_m}$ whenever $D$ has the prime decomposition in \Eq{eqn18}.

For brevity it is sufficient to prove the theorem for the case of $D = d_1 d_2$ where $d_1 = p_1^{\epsilon_1}$ and $d_2 = p_2^{\epsilon_2} \cdots p_m^{\epsilon_m}$ because $d_2$ can subsequently be further decomposed by induction. The Chinese Remainder ring isomorphism map, $\phi:\ZZ_D \rightarrow \ZZ_{d_1} \otimes \ZZ_{d_2}$ is defined as
\begin{align}
\label{eqnA1}
\phi(a) = (\phi_1(a), \phi_2(a)) 
\end{align}
where
\begin{align}
\label{eqnA2}
\phi_i(a) := a \bmod d_i. 
\end{align}
The inverse map $\phi^{-1} : \ZZ_{d_1} \otimes \ZZ_{d_2} \rightarrow \ZZ_D$ is given by
\begin{align}
\label{eqnA3}
\phi^{-1}(a_1, a_2) := a_1 r_1 d_2 + a_2 r_2 d_1 \bmod D
\end{align}
with $r_i := (D/d_i)^{-1} \bmod d_i$ being constants that depend only on $d_1, d_2$ and not inputs $a_1, a_2$. Note that $r_i$ is always coprime to $d_i$ for  $i=1, 2$.

Next we show the mapping $\phi$ induces a unitary transformation from a Hilbert space of dimension $D$ to a tensor product space of qudits of dimensions $d_1, d_2$. We define the action of the unitary $\UC : \HC_D \rightarrow \HC_{d_1} \otimes \HC_{d_2}$ on the $D$-dimensional basis kets as
\begin{align}
\label{eqnA4}
\UC \ket{a} &= \ket{\phi_1(a)} \otimes \ket{\phi_2(a)} \qquad \text{for $a = 0,1,\ldots, D-1$.} 
\end{align}
with $\phi_i$'s defined in Eq.~\eqref{eqnA2}. The fact that $\phi$ is bijective guarantees that $\{ \ket{\phi_1(a)} \otimes \ket{\phi_2(a)} \}_{a=0}^{D-1}$ spans the space $\HC_{d_1} \otimes \HC_{d_2}$.

The result of conjugating Pauli operator $X$ by this unitary is
\begin{align}
\label{eqnA5}
\UC X \UC^\dag &= \sum_{a=0}^{D-1} \ket{\phi_1(a)} \bra{\phi_1(a+1)} \otimes \ket{\phi_2(a)} \bra{\phi_2(a+1)} \nonumber \\
&= \sum_{a=0}^{D-1} \ket{\phi_1(a)} \bra{\phi_1(a)+1} \otimes \ket{\phi_2(a)} \bra{\phi_2(a)+1}  \nonumber \\
&=\sum_{a_1=0}^{d_1-1} \sum_{a_2=0}^{d_2-1} \ket{a_1} \bra{a_1+1} \otimes \ket{a_2} \bra{a_2+1} = X_1 \otimes X_2.
\end{align}
In the last line we simply replaced the sum over all elements of $\ZZ_D$ with the sum over all elements in $\ZZ_{d_1} \otimes \ZZ_{d_2}$. Note the subscript on $X$ here denotes two different qudits, each with different dimension.

Next the $Z$ operator is transformed as
\begin{align}
\label{eqnA6}
\UC Z \UC^\dag &= \sum_{a=0}^{D-1} \omega^a \ket{\phi_1(a)} \bra{\phi_1(a)} \otimes \ket{\phi_2(a)} \bra{\phi_2(a)} \nonumber \\
&= \sum_{a_i=0}^{d_i-1} \omega^{\phi^{-1}(a_1, a_2)} \ket{a_1} \bra{a_1} \otimes \ket{a_2} \bra{a_2} \nonumber \\
&= \bigotimes_{i=1}^2 \sum_{a_i=0}^{d_i-1} \mathrm{e}^{2  \pi \ii a_i r_i  /d_i} \ket{a_i} \bra{a_i} = Z_1^{r_1} \otimes Z_2^{r_2},
\end{align}
where $r_i$'s are defined in Eq.~\eqref{eqnA3}.

We are now ready to prove the theorem. Let $\AC$ be generated by $k$ elements, $\AC = \langle g^{(1)}, \ldots, g^{(k)} \rangle$ and the generators can always be chosen such that
\begin{align}
\label{eqnA7}
\prod_{i=1}^k \text{order} (g^{(i)}) = |\AC|,
\end{align}
where the order of every $g_i$ must be a divisor of $D$, see Subsection \ref{subsection2a} for our nonstandard definition of order. The requirement that $\AC$ be a collection of linearly independent Pauli products implies every generator satisfies $g_i^{\text{order} (g^{(i)})} = I$.

Consider an arbitrary generator, $g^{(i)}$ and define $\delta = \text{order} (g^{(i)})$. Let $\delta = p_1^{\xi_1} p_2^{\xi_2} \cdots p_m^{\xi_m}$ be the prime decomposition of $\delta$ with the same $p_i$'s as \Eq{eqn18}, and set $\delta_1 = p_1^{\xi_1}, \; \delta_2 = p_2^{\xi_2} \cdots p_m^{\xi_m}$, so $\delta = \delta_1 \delta_2$. Then $\delta_2$ is coprime to $\delta_1$ and also to $p_1$. Next define the unitary $\mathbb{U} = \UC \otimes \cdots \otimes \UC$ from Eq.~\eqref{eqnA4} acting on all the $n$ qudits. By Eqs.\eqref{eqnA5} and \eqref{eqnA6}, conjugating the first generator with $\mathbb{U}$ produces
\begin{align}
\label{eqnA8}
\mathbb{U} g^{(i)} \mathbb{U}^\dag = q_1 \otimes q_2
\end{align}
where $q_1, q_2$ are Pauli products on $n$ qudits of dimension $d_1, d_2$ respectively. Next we claim that
\begin{align}
\label{eqnA9}
\text{order} (q_1) = \delta_1 \quad \text{and} \quad \text{order} (q_2) = \delta_2.
\end{align}
To prove this, recall that $(g^{(i)})^{\delta_1 \delta_2} = I$ which implies $ q_1^{\delta_1 \delta_2} \propto I_1$. Since $\delta_2$ is coprime to $\delta_1$ and also to $d_1$, the unique multiplicative inverse of $\delta_2 \mod d_1$ exists. Then it follows that $q_1^{\delta_1 \delta_2 \delta_2^{-1}} = q_1^{\delta_1} \propto I_1$. Therefore, the order of $q_1$ must be a divisor of $\delta_1$ and by the same reasoning, the order of $q_2$ must be a divisor of $\delta_2$. The orders cannot be less than $\delta_1, \delta_2$ respectively as that would imply the order of $g^{(i)}$ is less than $\delta$.

Having proven \Eq{eqnA9}, we define  Pauli products of qudits of dimension $d_1$ and $d_2$, $h_1^{(i)}$ and $h_2^{(i)}$ respectively as
\begin{align}
\label{eqnA10}
I_1 \otimes h_2^{(i)} &:= \left( \mathbb{U} g^{(i)} \mathbb{U}^\dag \right)^{\mu_2 \delta_1} \nonumber \\
h_1^{(i)} \otimes I_2 &:= \left( \mathbb{U} g^{(i)} \mathbb{U}^\dag \right)^{\mu_1 \delta_2} 
\end{align}
where $\mu_i := (\delta/\delta_i)^{-1} \bmod \delta_i$. Next for all $\alpha_1 \in \ZZ_{\delta_1}$ and $\alpha_2 \in \ZZ_{\delta_2}$, we can rewrite \Eq{eqnA10} as
\begin{align}
\label{eqnA12}
\left( h_1^{(i)} \right)^{\alpha_1} \otimes \left( h_2^{(i)} \right)^{\alpha_2} &= \left( \mathbb{U} g^{(i)} \mathbb{U}^\dag \right)^{\alpha_1 \mu_1 \delta_2 + \alpha_2 \mu_2 \delta_1}.
\end{align}
Observe that the exponent on the right side is just the inverse Chinese Remainder map of $\alpha_1$ and $\alpha_2$ in\Eq{eqnA3} applied to this situation, so
\begin{align}
\label{eqnA13}
\left\langle \mathbb{U} g^{(i)} \mathbb{U}^\dag \right\rangle = \langle h^{(i)}_{1} \rangle \otimes \langle h^{(i)}_{2} \rangle.
\end{align}

Decomposing every generator into two generators on subsystems of dimensions $d_1$ and $d_2$ gives us
\begin{align}
\label{eqnA14}
\mathbb{U} \AC \mathbb{U}^\dag &= \langle \mathbb{U} g^{(1)} \mathbb{U}^\dag, \ldots, \mathbb{U} g^{(k)} \mathbb{U}^\dag \rangle   \nonumber \\
&= \langle h^{(1)}_{1}, \ldots, h^{(k)}_{1} \rangle \otimes \langle h^{(1)}_{2}, \ldots, h^{(k)}_{2} \rangle \nonumber \\
&= \AC_1 \otimes \AC_2.
\end{align}
That $\AC_1 = \langle h^{(1)}_{1}, \ldots, h^{(k)}_{1} \rangle$ and $\AC_2 = \langle h^{(1)}_{2}, \ldots, h^{(k)}_{2} \rangle$ form collections of mutually commuting Pauli products are consequences of the fact that $\{g^{(i)}\}$ are mutually commuting and that conjugation by $\mathbb{U}$ does not change the commutation relation. Finally, \Eq{eqnA14} also tells us there are as many linearly independent elements in $\AC$ as there are in $\AC_1 \otimes \AC_2$, which implies that $|\AC| = |\AC_1| \cdot |\AC_2|$.
\end{proof}

%==============================================================================

%\bibliography{tripart_bib}

\end{document}